\documentclass[aps,pra,reprint,bibnotes,floatfix]{revtex4-1}
\usepackage[latin9]{inputenc}
\usepackage{amsmath}
\usepackage{amssymb}
\usepackage{mathtools}
\usepackage{graphicx}
\usepackage[caption=false]{subfig}
\usepackage{braket}
\usepackage{hyperref}
\usepackage{paralist}
\usepackage{color}
\usepackage{mathrsfs}
\usepackage{amsthm}
\usepackage{centernot}
\usepackage{stmaryrd}
\usepackage{xcolor}
\usepackage{soul}
\usepackage{thmtools} 
\usepackage{thm-restate}
\usepackage{siunitx}

\usepackage{float}

\hypersetup{colorlinks=true,linkcolor=blue,citecolor=red,filecolor=red,urlcolor=blue,runcolor=blue}

\theoremstyle{definition}

\newtheorem*{proposition*}{Proposition}

\theoremstyle{plain}

\theoremstyle{definition}

\DeclareMathOperator*{\Tr}{Tr}

\DeclareMathOperator*{\Span}{Span}

\DeclareMathOperator{\inc}{inc}
\DeclareMathOperator{\rel}{rel}
\DeclareMathOperator{\HS}{HS}
\DeclareMathOperator{\unif}{unif}

\DeclareMathOperator{\PR}{PR}

\newcommand{\bs}{\boldsymbol}

\makeatletter

\begin{document}
\title{Quantum coherence and the localization transition}
%Measures of quantum coherence as probes to localization}

\date{\today}

\author{Georgios Styliaris}
\email [e-mail address: ]{styliari@usc.edu}
\author{Namit Anand}
\author{Lorenzo Campos Venuti}
\author{Paolo Zanardi}

\affiliation{Department of Physics and Astronomy, and Center for Quantum Information Science and Technology, University of Southern California, Los Angeles, California 90089-0484, USA}

\begin{abstract}

A dynamical signature of localization in quantum systems is the absence of transport which is governed by the amount of coherence that configuration space states possess with respect to the Hamiltonian eigenbasis. To make this observation precise, we study the localization transition via quantum coherence measures arising from the resource theory of coherence.
We show that the escape probability, which is known to show distinct behavior in the ergodic and localized phases, arises naturally as the average of a coherence measure. Moreover, using the theory of majorization, we argue that broad families of coherence measures can detect the uniformity of the transition matrix (between the Hamiltonian and configuration bases) and hence act as probes to localization. We provide supporting numerical evidence for Anderson and Many-Body Localization (MBL).

For infinitesimal perturbations of the Hamiltonian, the differential coherence defines an associated Riemannian metric. We show that the latter is exactly given by the dynamical conductivity, a quantity of experimental relevance which is known to have a distinctively different behavior in the ergodic and in the many-body localized phases. Our results suggest that the quantum information-theoretic notion of coherence and its associated geometrical structures can provide useful insights into the elusive nature of the ergodic-MBL transition.

%In this paper  we apply the quantum information-theoretic notion of coherence generating power (CGP) and its geometrical -both finite and differential - affiliates to study eigenstate quantum phase transitions. In particular we analyze the ergodic-Many Body Localized (MBL) transition in terms of the scaling  properties of the CGP associated to the adiabatic-intertwiner connecting Hamiltonians
%within different physical phases. We argue that in the ergodic (MBL) phase in the thermodynamical limit the CGP is maximal (sub-maximal). Moreover, for infinitesimal perturbations of the Hamiltonian one can resort to the differential approach to CGP and show that the associated Riemannian metric can be mapped onto well-known physical quantities which have a distinctively different behavior in the ergodic and in the MBL phases. 
\end{abstract}

\maketitle

\section{Introduction}

One of the conceptual pillars of quantum theory is the superposition principle and, directly arising from it, the notion of \textit{quantum coherence} \cite{bohm2012quantum}. A quantum state is deemed to be coherent with respect to a complete set of states if it can be expressed as a nontrivial linear superposition of these states.
%Although coherence is one of the oldest concepts in quantum mechanics, it was only recently that 
Recently, there has been an effort to formulate a resource theory of quantum coherence \cite{aberg2006quantifying,PhysRevLett.113.140401,RevModPhys.89.041003}. The focus of this theory has been quantum information processing tasks, since generating and preserving quantum coherence constitutes one of the essential prerequisites.

In this work, we utilize the powerful tools that arose from this information-theoretic perspective on coherence to study phase transitions in quantum one- and many-body systems. More specifically, we focus on Anderson \cite{anderson1958absence,lagendijk2009fifty} and Many-Body Localization (MBL) transitions \cite{basko_metal-insulator_2006,pal_many-body_2010,nandkishore_many-body_2015}. These ``infinite temperature'' or ``eigenstate'' phase transitions  are characterized by a abrupt change occurring at the level of whole Hamiltonian eigenstates as opposed, e.g., to the ground-state only. 
%A prominent example of these ``infinite temperature'' or ``eigenstate'' phase transitions is provided by the phenomena of Anderson \cite{anderson1958absence,lagendijk2009fifty} and Many-Body Localization (MBL) \cite{basko_metal-insulator_2006,pal_many-body_2010,nandkishore_many-body_2015}.
%The latter transitions are going to be the focus of the present work.

A connection between quantum coherence and the transition of a quantum system from an ergodic phase to a localized one can be conceptually formalized as follows. One of the signatures of localization is the absence of transport, with respect to some properly defined positional degree of freedom. On the other hand, transport properties are governed by the coherence between the Hamiltonian eigenbasis and the positional one. Hence one should expect an abrupt change in the coherence properties of the Hamiltonian eigenvectors at the transition point.

Here we make the above intuition quantitatively precise by investigating the amount of coherence that can be generated on average by the quantum dynamics starting from incoherent states, the \emph{Coherence-Generating Power} (CGP) of a quantum evolution.
%invoking measures of quantum coherence, defined with respect to the energy eigenbasis, averaged over a complete set of position states.
%following the approach developed in \cite{coherence_1,coherence_2,coherence_3}.
Such quantities essentially capture the difference between two complete orthonormal sets of eigenstates associated with two hermitian operators \cite{zanardi2018quantum,styliaris2019quantifying}. We first show that a well-studied quantity in localization, the \textit{escape probability} (or, equivalently, the \textit{second participation ratio}) can be expressed directly as a coherence average.  We then argue that broad families of coherence measures, arising from the the resource-theoretic perspective, can be used to define an ``order parameter'' for localization. We provide supporting numerical evidence for both Anderson and MBL transitions. Moreover, we show that the differential-geometric version of our average coherence is exactly given by an infinite temperature dynamical conductivity, an experimentally accessible quantity, which is known to behave differently in the ergodic and MBL phases \cite{Prelovsek2017MBL}. These findings open the possibility of observing experimentally the coherence generating power of quantum dynamics.    

This paper is organized as follows. In \autoref{Sec_prelim} we introduce measures of coherence for quantum states and explain how one can average coherence over a complete set of states in order to obtain the associated CGP, i.e., a quantifier of coherence for unitary quantum operations. We then investigate general mathematical properties of the aforementioned coherence averages and connect with the theory of matrix majorization and the escape probability. In \autoref{Sec_Anderson} we examine, both analytically and numerically, the behavior of two of the introduced measures in the Anderson localization transition and connect with the associated localization length. In \autoref{Sec_MBL} we numerically study the introduced measures for a many-body system that exhibits a transition to a MBL phase. In \autoref{sec_diffgeom} we turn to the Riemannian metric that results from the average coherence between bases that differ infinitesimally and relate with the MBL transition. Finally, in \autoref{sec_conclusions} we conclude with a discussion and future work. All proofs can be found in section \ref{Sec_app_proofs} of the Appendix.

%These findings  suggest that the quantum information-theoretic notion of coherence and its associated geometrical structures may provide useful insights  into  the nature of the elusive ergodic-MBL transition.

\section{Quantum coherence of states and operations} \label{Sec_prelim}

\subsection{Coherence of states}

Consider a quantum system, described by a finite dimensional Hilbert space $\mathcal H \cong \mathbb C ^d$.
%Quantum coherence refers to the capability of physical states to be expressible as linear combinations in terms of other states.
A state $\ket{\psi} \in \mathcal H$ is deemed \textit{coherent} with respect to a fiducial orthonormal basis $\left\{ \ket{\phi_i} \right\}_{i=1}^{d}$ if the expansion $\ket{\psi} = \sum_i a_i \ket{\phi_i}$ contains more than one non-vanishing term, otherwise it is called \textit{incoherent}. This notion extends straightforwardly to the set of density operators $\mathcal S(\mathcal H)$. Any $\rho \in \mathcal S(\mathcal H)$ is regarded as \textit{coherent} with respect to the preferred basis if the corresponding matrix $\rho_{ij}$ has nonzero off-diagonal elements, otherwise it is termed \textit{incoherent}.

Quantum coherence is usually defined relative to a reference basis. In fact, one needs a weaker notion than that of a basis, since phase degrees of freedom and ordering of an orthonormal basis $\left\{ \ket{\phi_i} \right\}_{i=1}^{d}$ are physically redundant. In other terms, bases differing by transformations of the form $\ket {\phi_j} \mapsto e^{i \theta_j} \ket{\phi_{\pi (j)}}$ ($\pi \in  S_{d}$ is a permutation)
%\NA{use \(S_d\) instead of \(\mathcal S_{d}\) for the permutation group}
are equivalent as far as coherence is concerned. The relevant object, taking into account this freedom, is a complete set of orthogonal, rank-1 projection operators $B=\left\{\Pi_i\right\}_{i=1}^{d}$, where $\Pi_i \coloneqq \ket{\phi_i}\!\bra{\phi_i}$. In the rest of this work, we will refer for convenience to the set $B$ itself as a ``basis''.

While all states non-diagonal in $B$ carry coherence, some of them might resemble incoherent states more than others. This notion is made precise by the introduction of ($B$-dependent) functionals, $c_{B}: \mathcal S(\mathcal H) \to \mathbb R_0^+$ that are said to quantify coherence \cite{PhysRevLett.113.140401}. Quantifiers of coherence (also called \textit{coherence monotones}) satisfy $c_B(\rho_{\inc}) = 0$ for all states diagonal in $B$ and, in addition, are non-increasing under the free operations of the resource theory \footnote{We note that there exist various proposals for the free operations in the resource theories of coherence (see \cite{PhysRevA.94.052336} for more details). In the following, we will use the term Incoherent Operations for the free operations but, in fact, all results hold for any class that contains Strictly Incoherent Operations \cite{PhysRevLett.116.120404,PhysRevX.6.041028}.}.
%Possible sets of axioms have been established in the framework of coherence as a resource theory, where coherence quantifiers jointly capture the conversion relations among quantum states under the action of some distinguished set of operators that cannot create coherence out of incoherent states. 
In this work, we make use of the \textit{2-coherence} and the \textit{relative entropy of coherence}, defined respectively by
\begin{subequations}
\begin{align}
c^{(2)}_B (\rho) &\coloneqq \left\| \left( \mathcal I -\mathcal D_{B} \right) \rho  \right\|^2_2  = \sum_{i \ne j} \left| \rho_{ij} \right|^2  \\
c^{(\rel)}_B (\rho) &\coloneqq S \left[ \mathcal D_B  \left( \rho \right)  \right]  - S \left( \rho \right) \,\;,
\end{align}
\end{subequations}
where we have introduced the $B$-dephasing superoperator
\begin{align} \label{Eq_dephasing_def}
\mathcal D_B (X) \coloneqq \sum_{i=1}^{d}  \Pi_i X \Pi_i \,\;,
%\braket{i|X|i} \ket{i}\!\bra{i} \,\;.
\end{align}
$S$ above denotes the usual von-Neumann entropy $S(\rho) \coloneqq -\Tr \left(  \rho \log (\rho)\right)$ and the (Schatten) 2-norm of an operator $X$ is defined as $\left\| X \right\|_2 \coloneqq \sqrt{\Tr\left(X^\dagger X\right)}$. Relative entropy of coherence is a central measure in the resource theories of coherence and admits an operational interpretation, e.g.,  as a conversion rate of information-theoretic protocols \cite{PhysRevLett.116.120404,PhysRevLett.120.070403}. The 2-coherence admits an interpretation as an escape probability, as will be shown momentarily \footnote{We note, however, that the 2-coherence might fail to satisfy the monotonicity property under some classes of free operations.}.

\subsection{Coherence of unitary quantum processes via probabilistic averages}

In this section we discuss how, given a coherence measure $c_B$ and a unitary superoperator $\mathcal U$, one can capture the ability of the unitary $\mathcal U$ to generate coherence by computing the average amount of coherence that can be generated starting from incoherent states. This is the \textit{Coherence-Generating Power} (CGP) of the quantum operation $\mathcal U$ \cite{coherence_1,coherence_2,coherence_3,zanardi2018quantum}. 

%In this section we discuss how, given a coherence measure $c_B$, one can define a corresponding notion of coherence for a unitary superoperator $\mathcal U$ via probabilistic averages. We refer to this approach to quantification of coherence for quantum operations as \textit{Coherence-Generating Power} (CGP) \cite{coherence_1,coherence_2,coherence_3,zanardi2018quantum}. 

%This approach to \NA{(quantifying the)} coherence of quantum operations, termed \textit{Coherence-Generating Power} (CGP), was proposed and investigated (for the more general case of unital operations) in Refs.~\cite{coherence_1,coherence_2,coherence_3,zanardi2018quantum}. 

Consider a basis $B = \left\{ \Pi_i \right\}_{i=1}^{d}$ and define a probabilistic ensemble of incoherent states, i.e., a random variable $\rho_{\inc} (\bs p) = \sum_i p_i \Pi_i$, where $\left\{ p_i \right\}_i$ $(p_i \ge 0$, $\sum_{i}p_i = 1)$ are random and distributed according to a prescribed measure $\mu(\bs p)$. Then, the corresponding CGP
\begin{align}
  C \left( \mathcal U , c_B , \mu \right) \coloneqq \int d \mu(\bs p) \, c_B \left[ \mathcal U \left( \rho_{\inc} (\bs p) \right) \right]
\end{align}
characterizes the average effectiveness of the quantum process $\mathcal U$ to generate coherence out of random incoherent states in $B$. Since the unitary $\mathcal U(X) = U  X U^\dagger$ can be thought of as connecting the bases $B$ and $B' = \{ \mathcal U \left( \Pi_i \right) \}_i$, one can also interpret $C \left( \mathcal U , c_B , \mu \right) $ as the average coherence with respect to $B$ of a random state which is incoherent in $B'$.

Without any additional structure, it is a natural choice to consider averaging only over pure states with equal weight over each of them , i.e., take

%For the purposes of this work, it is convenient to average only over pure states with equal weight over each of them, i.e., take
\begin{align}
\mu_{\unif}(\bs p) \coloneqq \frac{1}{d} \sum_i \delta\left( \bs p - \bs e_i \right) 
\end{align}
where $(\bs e_i)_j \coloneqq \delta_{ij}$ \footnote{In \cite{coherence_1}, the measure considered was the uniform over the (whole) simplex of probability distributions, instead of just the extremal ones, resulting in an extra factor ${(d+1)^{-1}}$.}.
This choice directly leads to the expression
%for the average coherence generated by a unitary process:
\begin{align}  
	C \left( \mathcal U , c_B ,\mu _{\unif} \right) = \frac{1}{d} \sum_{i=1}^{d} c_B \left[ \mathcal U \left(  \Pi_i \right) \right] \,\;. \label{Eq_CGP_general}
\end{align}

We now simplify Eq.~\eqref{Eq_CGP_general} when the coherence measure is the 2-coherence or the relative entropy of coherence, namely for
%the CGP cases
\begin{subequations}
\begin{align}
C^{(2)}_B \left( \mathcal U \right) &\coloneqq C \left( \mathcal U , c^{(2)}_B ,\mu _{\unif} \right) \,\;, \\
C^{(\rel)}_B \left( \mathcal U \right) &\coloneqq C \left( \mathcal U , c^{(\rel)}_B ,\mu _{\unif} \right) \,\;.
\end{align}
\end{subequations}

\begin{restatable}{prop}{cgpformulas}
\label{Prop_CGP_formulas}
Let $B = \left\{ \Pi_i \right\}_{i=1}^{d}$ be a basis, $\mathcal U$ a unitary quantum process and $X_{\mathcal U}$ denote the bistochastic matrix with elements $\left( X_{\mathcal U}\right)_{ij} \coloneqq \Tr \left( \Pi_i \, \mathcal U(\Pi_j) \right)$. Then,
\begin{align} \label{Eq_2CGP_Tr}
C^{(2)}_B \left( \mathcal U \right)  = 1 - \frac{1}{d} \Tr \left( X_{\mathcal U} ^T X_{\mathcal U}  \right) \,\;.
\end{align}
and
\begin{align}\label{Eq_relCGP_H}
C^{(\rel)}_B \left( \mathcal U \right) = H(X_{\mathcal U}) \,\;,
\end{align}
where $H(X) \coloneqq - \frac{1}{d} \sum_{i,j} X_{ij} \log (X_{ij})$ denotes the generalization of the Shannon entropy over bistochastic matrices.
\end{restatable}

The two CGP quantities are related as
\begin{align}
    C_B ^{(\rel)} \ge -\log \left( 1 - C_B^{(2)} \right) \,\;.
%    1 - C_B^{(2)} \ge e^{ - C_B ^{(\rel)} } \,\;.
\end{align}
The inequality follows from the above Proposition, together with the
%use of $H(\bs p) \ge -\log \left( \sum_i p_i ^2 \right)$ (arising from the monotonicity of the R\'{e}nyi entropies \cite{cover2012elements}) and the
concavity of the logarithmic function.

\subsection{General properties of coherence-generating power measures}

Both quantities $C_B^{(2)} (\mathcal U)$ and $C_B^{(\rel)} (\mathcal U)$ introduced earlier can be considered as functions of the (transition) matrix $X_{\mathcal U}$, instead of $\mathcal U$ itself. In other words, the phases associated with $U_{ij}$ (treated as a matrix in the $B$ basis) are irrelevant. In fact, as we will show momentarily, this is a general feature of any CGP measure $C \left( \mathcal U , c_B ,\mu _{\unif} \right)$ arising from a coherence monotone $c_B$.

%Indeed, consider a pure state in $B' = \{ \mathcal U \left( P_i \right) \}_i$. Then, the value of $c_B\left[  \mathcal U \left( P_i \right) \right]$ can only depend on the coefficients

%Consider a pure state $\ket{\psi}$. The value of $c_B\left(  \ket{\psi} \! \bra{\psi} \right)$ can only depend on the set of coefficients $\left\{ \left| \braket{i | \psi } \right|  \right\}_{i=1}^d$ \footnote{This follows from the fact that any unitary transformation that changes the phases or permutes the elements of $B$ forms a subgroup of the Incoherent Operations, hence all coherence monotones should maintain a constant value over the group orbit.}. As a result, all phases associated with the coefficients of $\ket{\psi}$ (represented as a vector in the $B$ basis) are irrelevant as far as coherence is concerned. As a result, the CGP $C \left( \mathcal U , c_B ,\mu _{\unif} \right)$ does not depend on the phases associated with the pure states $\{ \mathcal U \left( P_i \right) \}_i$, i.e., can be considered as a function of $X_{\mathcal U}$.

Motivated by the above observation, we define as a \textit{generalized CGP measure} any function $f_B$ mapping bistochastic matrices to non-negative real numbers such that:
\begin{enumerate}[(i)]
\item $f_B \left( \Pi  \right) = 0$ if $\Pi \in S_d$ is a permutation.
\item $ f_B \left(\Pi  X  \Pi' \right)  =  f_B \left( X  \right) $, where $\Pi,\Pi' \in S_d$ are permutations.
\item $f_B \left(M  X  \right) \ge f_B \left(X  \right) $ for any bistochastic matrix $M$.
\end{enumerate}

\begin{restatable}{prop}{cgpproperties}
\label{Prop_CGP_properties}
Let $c_B$ be a coherence measure. Then, the corresponding Coherence-Generating Power $f_B (X_{\mathcal U}) \coloneqq C \left( \mathcal U , c_B ,\mu _{\unif} \right)$ satisfies $\mathrm{(i) - (iii)}$ above.
\end{restatable}

%It can be easily shown that the above properties are satisfied by any CGP $C \left( \mathcal U , c_B ,\mu _{\unif} \right)$ (understood as a function of $\mathcal U$), given that $c_B$ is a coherence measure.

%We now construct two examples of generalized CGP measures.

On physical grounds, all quantities $C \left( \mathcal U , c_B ,\mu _{\unif} \right)$ are expected to quantify how ``uniform'' or ``spread'' is the transition matrix $X_{\mathcal U}$ between the bases $B$ and $B' = \mathcal U(B)$.
%which in our previous examples  were taken the be the energy and position eigenbases.
This intuition is reflected in part (iii) of \autoref{Prop_CGP_properties}: ``post-processing'' the transition matrix $X \mapsto MX$ by any bistochastic matrix $M$ will certainly increase any CGP measure $C \left( \mathcal U , c_B ,\mu _{\unif} \right)$, where $c_B$ can be any coherence monotone. 

Generalized CGP measures can be thought of as functions that characterize the uniformity of a (bistochastic) matrix. They always achieve their maximum value over the transition matrix $\left( X_{\mathcal V} \right)_{ij} = 1/d$, i.e., when $\mathcal V$ connects two unbiased bases, as follows by combining properties $(ii)$ and $(iii)$. In a similar manner, the minimum value is achieved over permutation matrices and is set to zero (as a normalization) by $(i)$. For instance, any concave function that satisfies properties $(i)$ and $(ii)$ automatically satisfies $(iii)$, i.e., is a generalized CGP measure.

Examples of generalized measures arising from previous works on CGP (see \cite{coherence_2,zanardi_quantum_2018,styliaris2019quantifying}) are
\begin{align} \label{Eq_det}
    f_B^{(\det)} \left( X_{\mathcal V}  \right) & \coloneqq 1 - \left| \det \left( X_{\mathcal V}  \right) \right| ^{\frac{1}{d}} \\
    f_B^{(\infty)} \left( X_{\mathcal V}  \right) & \coloneqq \left\|  I - X_{\mathcal V}^T X_{\mathcal V} \right\|_\infty \,\;,
\end{align}
where $\left\| \left(  \cdot  \right) \right\|_\infty$ denotes the operator norm. Notice that $f_B^{(\det)} (X_{\mathcal V}) = 1 - \left( \prod_i s_i \right)^{\frac{1}{d}} $ and also $0 \le f_B^{(\det)} \left( X_{\mathcal V}  \right) \le 1$, while $f_B^{(\infty)} \left( X_{\mathcal V}  \right) = 1 - s_d^2$ (here $s_i$ are the singular values of $X_{\mathcal V}  $ sorted in decreasing order).

%In the previous sections, we illustrated a connection between localization and a change in the behavior of various CGP measures via numerical simulations.

A systematic way to capture the amount of uniformity of a matrix is provided by the notion of multivariate majorization \cite{marshall_inequalities:_2011}. An example is \textit{column majorization}, in which a stochastic matrix $X$ column majorizes another stochastic matrix $Y$, denoted as $X \succ^c Y$, if $X^c_i \succ Y^c_i$ $\forall i$; here $X^c_i$ and $Y^c_i$ stand for the $i^{\text{th}}$ column vector of $X$ and $Y$, respectively, and ``$\succ$'' denotes ordinary majorization of probability vectors.

It is then natural to ask whether the CGP quantities $C \left( \mathcal U , c_B ,\mu _{\unif} \right)$ arising from different coherence measures $c_B$ jointly capture some notion of uniformity of the transition matrix $X_{\mathcal U}$, as described by multivariate majorization. We answer this in the affirmative via the proposition below. 

\begin{restatable}{prop}{cgpmajorization}
\label{Prop_majorization}

Let $c_B$ be a coherence measure. Then, the corresponding Coherence-Generating Power $f_B (X_{\mathcal U}) \coloneqq C \left( \mathcal U , c_B ,\mu _{\unif} \right)$ considered over bistochastic matrices is a monotone of column majorization, i.e., $X \succ^c Y$ $\Rightarrow$ $f_B(X) \le f_B(Y)$. Conversely, if $f_B(X) \le f_B(Y)$ for all $f_B$ arising from continuous coherence monotones over pure states, then $X \succ^c Y$.

\end{restatable}

The last part of the above proposition establishes the fact that there are enough coherence monotones over pure states one can consider such that, if all corresponding measures $f_B$ are monotonic, then column majorization is guaranteed. In other words, these functions form a complete set of monotones. In that sense, the defined family of CGP measures jointly captures a notion of uniformity for the transition matrix that is at least as strict as column majorization.

\subsection{Coherence and escape probability}

%In this section we establish a connection of our previous considerations on coherence, motivated by an information-theoretic approach, with dynamical properties of quantum systems.
Let us consider  a finite-dimensional quantum system whose dynamics is specified by a Hamiltonian $H$. Suppose the system is initialized in a state $\ket{\psi}$ and one is interested in the \textit{escape probability} %\footnote{The second term is just the time-average of the Loschmidt echo.}
\begin{align}
\mathcal P_\psi \coloneqq  1 - \overline{ \left| \braket{\psi | e^{-i H t} | \psi} \right|^2 } \,\;,
\end{align}
where the overline denotes the infinite time-average
\begin{align}
\overline{f(t)} \coloneqq \lim_{T \to \infty} \frac{1}{T} \int_{0}^T dt \, f(t) \,\;.
\end{align}
For instance, in the case of a particle hopping on a lattice which is initialized over a single site $j$, $\mathcal P_j$ corresponds to the average probability of the particle escaping the initial site.

%Notice that, for a non-degenerate Hamiltonian $H = \sum_i E_i \ket{\phi_i} \bra{\phi_i}$, the escape probability $\mathcal P_\psi$ is directly connected with the second Participation Ratio of $\ket{\psi}$ over the Hamiltonian eigenbasis $\PR_2  \coloneqq \sum_i \left| \braket{\phi_i | \psi} \right|^4 $ as $\mathcal P_\psi = 1 - \PR_2$.

%For a non-degenerate Hamiltonian, the second Participation ratio is equal to the inverse \textit{effective dimension}, which is defined as the purity of the (infinite) time-averaged state.

%If the Hamiltonian in consideration further satisfies the \textit{non-degenerate energy gaps} condition,

At this point, let us note that in finite dimensions observable quantities such as $\braket{ A(t)} \coloneqq \Tr[A(t) \rho_{0}]=\Tr[A \rho(t)]$
%for initial state $\rho_{0}$ and time-evolved observable $A(t)$
do not converge to any limit as $t\to\infty$. Instead they start from an initial value
and then oscillate around a value given by $\overline{\braket{A(t)}}$
\cite{PhysRevLett.101.190403,PhysRevE.79.061103,PhysRevA.81.022113,PhysRevLett.107.010403}.
Since if a function $f(t)$ has a limit for $t\to\infty$, this
limit must coincide with $\overline{f(t)}$, the infinite
time average provides a way to extract the infinite time limit even
when the latter strictly speaking does not exist. 

If the Hamiltonian in consideration has non-degenerate energy gaps (also known as the non-resonance condition),
%(which is a stronger constraint than non-degeneracy),
the \textit{effective dimension} $d_\mathrm{eff} \coloneqq (1 - \mathcal P_\psi)^{-1}$ %(also known as \textit{effective dimension}) 
dictates the equilibration properties of the system: the larger $d_\mathrm{eff}$  the smaller are the temporal fluctuations of the observables around their mean values~\cite{PhysRevE.79.061103,PhysRevLett.101.190403}, i.e., equilibration is stronger. Since many-body localization is a mechanism by which quantum systems can escape equilibration, it is perhaps no surprise that the effective dimension is related to the localization transition (see the Appendix~\ref{Sec_App_thermalization} for more details on related quantities).

After introducing the basic framework, we are now ready to present our first result. The following Proposition establishes the fact that the 2-coherence of a state, quantified with respect to the Hamiltonian eigenbasis, is the time-averaged escape probability of the state.

%provides a physical interpretation of the coherence measure $c^{(2)}_B$ over pure states and the corresponding CGP measure $C^{(2)}_B$ in terms of the escape probability.

\begin{restatable}{prop}{cgpmeaning}
\label{Prop_CGP_meaning}
Let $H = \sum_i E_i \ket{\phi_i}\!\bra{\phi_i}$ be a non-degenerate Hamiltonian.
\begin{enumerate}[(i)]
\item For any state $\ket \psi$,
\begin{align}
\mathcal P_\psi = c_{B}^{(2)}(\ket{\psi}\!\bra{\psi}) \,\;,
\end{align}
where $B = \left\{ \ket{\phi_i}\!\bra{\phi_i} \right\}_i$ is the eigenbasis of the Hamiltonian.
\item Denote the escape probability averaged over a set of orthonormal states $B'=\left\{ \ket{i}\!\bra{i} \right\}_{i=1}^{d}$ as
\begin{align} \label{Eq_return_prob_CGP}
\mathcal P_{B'} \coloneqq \frac{1}{d} \sum_{i=1}^{d} \mathcal P_i \,\;.
\end{align}
Then,
\begin{align}  \label{Eq_return_prob_CGP_2}
\mathcal P_{B'}  = C_{B}^{(2)}(\mathcal V)  = C_{B'}^{(2)}(\mathcal V ^\dagger)  \,\;, 
\end{align}
where $B = \left\{ \ket{\phi_i}\!\bra{\phi_i} \right\}_i$ is the eigenbasis of the Hamiltonian and $\mathcal{V}(\cdot)\coloneqq V \left(\cdot\right) V^\dagger$, where $V = \sum_{i} \ket{i} \! \bra{\phi_i}$  is the intertwiner between $B$ and $B'$.
\end{enumerate}
\end{restatable}

The last equation above demonstrates that the role of the bases $B$ and $B'$ can be interchanged. For instance, one can equivalently think in terms of the average coherence over Hamiltonian eigenstates, quantified with respect to the position basis. 

A physically relevant family of unitary transformations $\mathcal U_t$ is the time evolution generated by the Hamiltonian of a system. One can, for instance, consider the time-average of $C_{B'}^{(2)}(\mathcal U_t)$. For a Hamiltonian with non-degenerate energy gaps, the aforementioned quantity admits the closed form expression
\begin{align} \label{Eq_CGP_NRC}
    \overline{C_{B'}^{(2)} (\mathcal U_t)} = 
    1 - \frac{2}{d} \sum_{ij} \braket{X^c_i , X^c_j}^2 + \frac{1}{d} \sum_{i} \braket{X^c_i , X^c_i}^2 \,\;,
\end{align}
here $X^c_i$ stands for the column vector of the transition matrix $X_{\mathcal{V}}$, while $V = \sum_{i} \ket{i} \! \bra{\phi_i}$ is the intertwiner between the Hamiltonian eigenbasis $B = \{ \ket{\phi_i}\!\bra{\phi_i}  \}$ and $B' = \{ \ket{i}\!\bra{i} \}_i$. In fact, the resulting quantity $f_B^{(\text{time-avg})} (X_{\mathcal V}) \coloneqq \overline {C_{B'}^{(2)} (\mathcal U_t)  }$ fails to be a generalized CGP measure. The details can be found in \autoref{Sec_App_time_avg} of the Appendix.

The identification between return probability and 2-coherence gives a physical interpretation to the latter and its associated CGP.
%while the second part demonstrates the interchangeable roles of the two bases under averaging.
More importantly, the escape (or return) probability is a well-known measure in the theory of localization \cite{anderson_absence_1958,Luitz_Lev_review_2017} and the fact that it can be thought of as coherence gives rise to the question: can other measures arising from the resource theoretic framework of coherence give rise to probes of localization in a similar manner? In what follows, we demonstrate that this is indeed the case, by considering Anderson and MBL.
%, where we study the two introduced CGP quantities.

\section{Coherence-generating power and localization in the 1-D Anderson model} \label{Sec_Anderson}

%In this section, we study the behavior of the CGP measures introduced earlier over the Anderson model \cite{anderson_absence_1958}.
The Anderson model \cite{anderson_absence_1958} in one dimension is described by the Hamiltonian
\begin{align}
H_W = - \sum_{i=1}^L \left( \ket{i}\!\bra{i+1} +  \ket{i+1}\!\bra{i} \right) + \sum_{i=1}^L \epsilon_i  \ket{i}\!\bra{i} \label{Eq_Anderson_Hamiltonian}
\end{align}
over $L$ sites (i.e., $d = L$) with periodic boundary conditions, where the on-site energies $\epsilon_i$ are independent and identically distributed (i.i.d.)~random variables and follow a uniform distribution of width $2W$. It is known that the model is localized for any degree of disorder $W>0$ \cite{hundertmark2008short}.

\begin{figure}[t]

\subfloat[average return probability (Anderson)]{%
  \includegraphics[clip,width=\columnwidth]{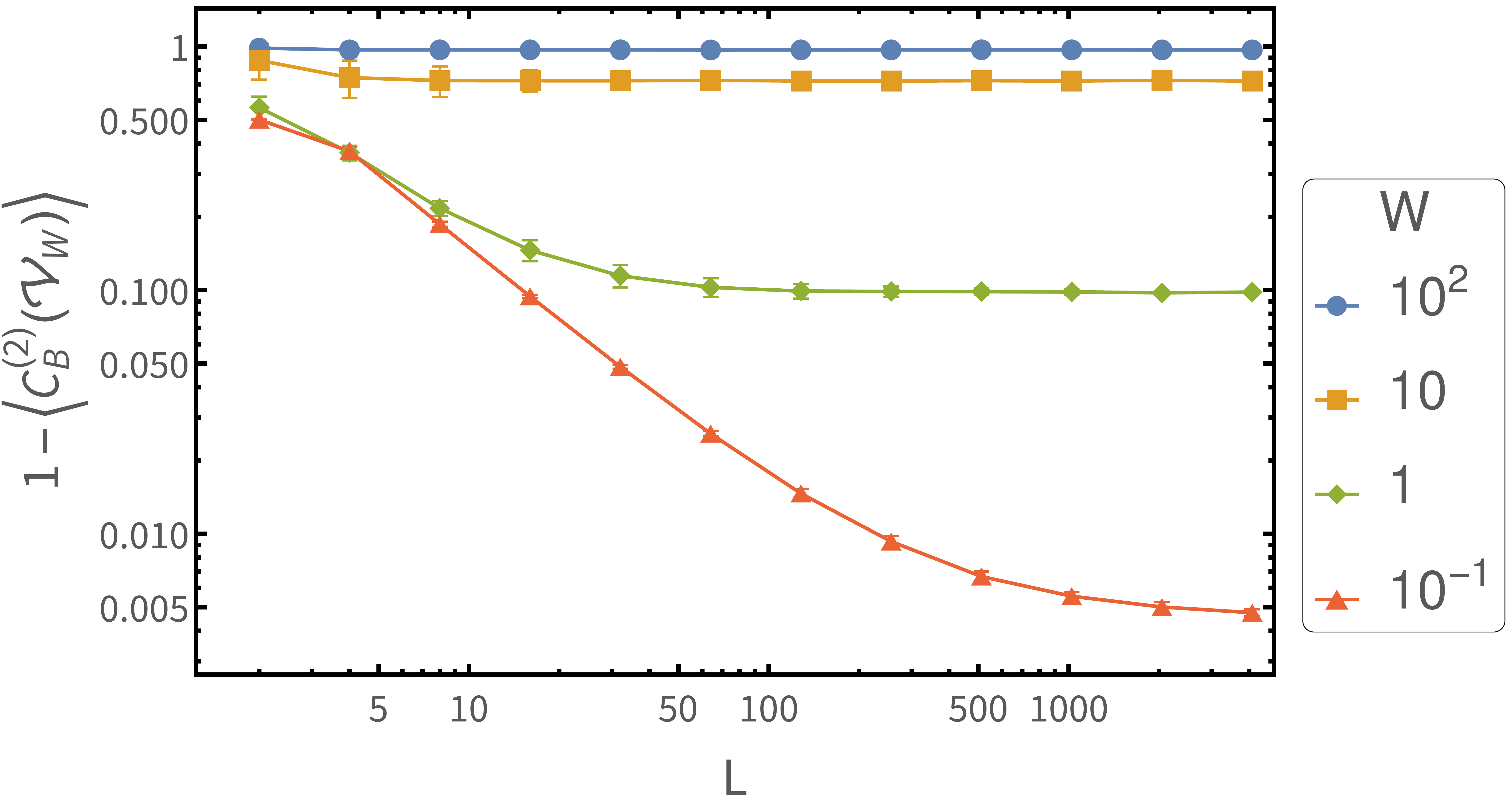}%
}

\subfloat[relative entropy CGP (Anderson)]{%
  \includegraphics[clip,width=.94 \columnwidth]{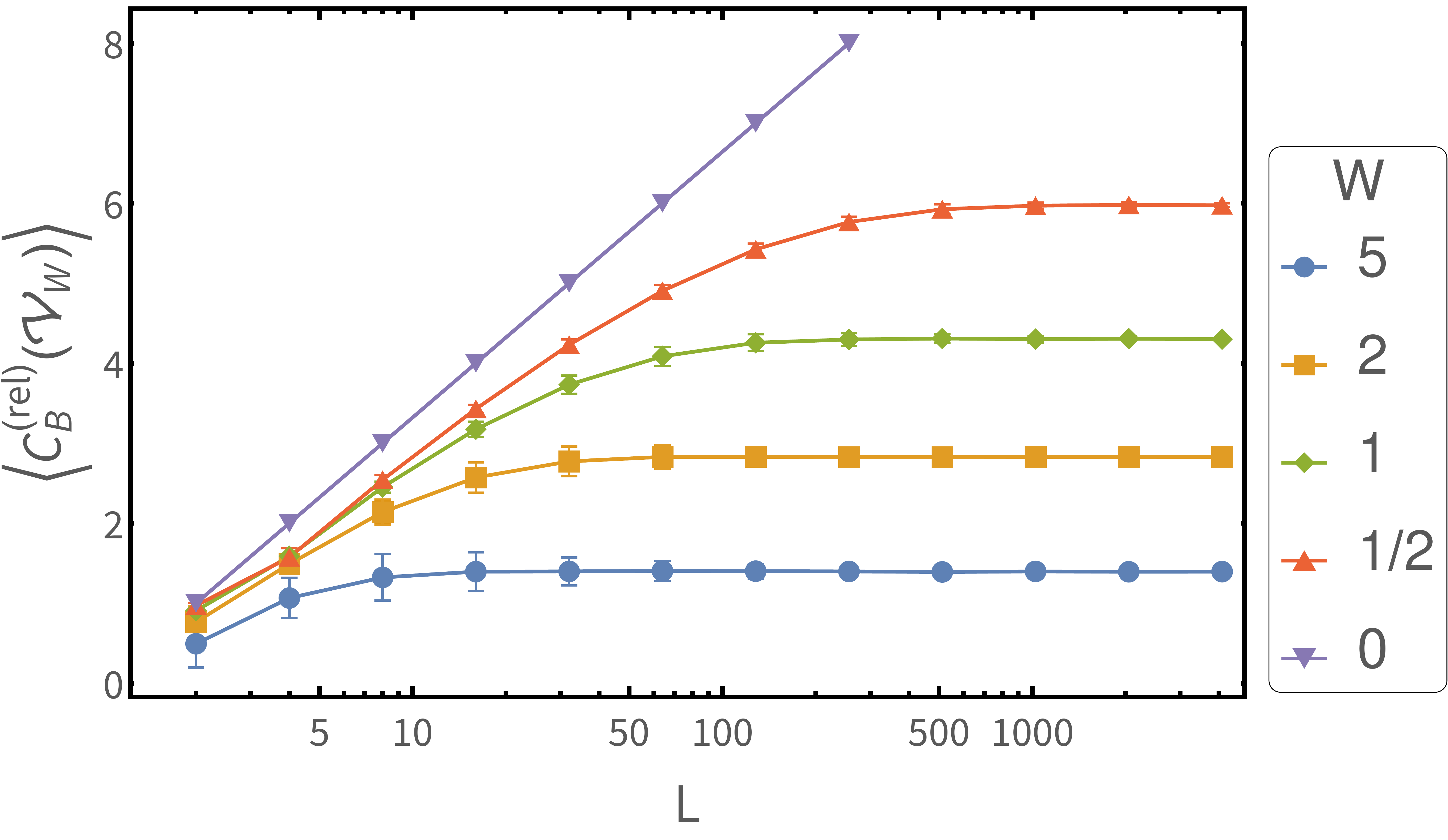}%
}

\caption{\textbf{(a)} Log-log plot of the average return probability $1 - \braket{ C_B^{(2)} \left( \mathcal V_W \right) }$ as a function of the system size $L$ for different values of the disorder strength $W$. The system is in the localized phase for all $W>0$, since the asymptotic escape probability is strictly less that 1 for $L \to \infty$. \\
\textbf{(b)} Log-linear plot of $\braket{ C_B^{(\rel)} \left( \mathcal V_W \right) }$ as a function of the system size $L$ for different values of the disorder strength $W$. The system is in the localized phase for all $W>0$, in which the asymptotic value is finite. In the ergodic phase ($W = 0$) $\braket{ C_B^{(\rel)} \left( \mathcal V_W \right) }$ diverges logarithmically. \\
The number of realizations range from \num{30000} for small sizes to just 8 for the largest size. Error bars represent one standard deviation. Entropy has logarithm with base 2.} \label{Fig_arp}

\end{figure}

%\begin{figure}[t]
%\centering
%\includegraphics[width=0.45\textwidth]{arp.pdf}
%\caption{Log-log plot of the average return probability $1 - \braket{ C_B^{(2)} \left( \mathcal V_W \right) }$ as a function of the system size $L$ for different values of the disorder strength $W$. The system is in the localized phase for all $W>0$, since the asymptotic escape probability is strictly less that 1 for $L \to \infty$. The number of realizations range from \num{30000} for small sizes to just 8 for the largest size. Error bars represent standard deviations.} \label{Fig_arp}
%\end{figure}

Localization can be dynamically characterized by the absence of transport, a notion referring to the interplay between the ``position'' basis $B' = \left\{ \ket {i}\!\bra{i} \right\}_{i=1}^L$ in Eq.~\eqref{Eq_Anderson_Hamiltonian} and the Hamiltonian eigenbasis $B$. Here, we consider coherence quantified with respect to the latter basis. Let us now examine the behavior of functionals $C_B(\mathcal V_W)$, where the unitary $V_W$ is the intertwiner between Hamiltonian and position eigenbases. In fact, \autoref{Prop_CGP_meaning} immediately implies that $\braket{C_B^{(2)}(\mathcal V_W)}$
%where $V$ is the intertwiner between the position and the Hamiltonian eigenbasis
is a probe to localization ($ \braket{\cdot} $ denotes averaging over disorder). More specifically, localization implies that in thermodynamic limit the return probability (averaged over disorder) in the localized phase is non-vanishing, i.e., $\lim_{L \to \infty} \braket{ \overline{\left| \braket{j | e^{-i H_W t} | j} \right|^2} } > 0$ for any $W>0$. In turn, this is equivalent to $\mathcal P_j < 1$ (in the thermodynamic limit) for all sites $j$, hence also 
\begin{align}
\lim_{L \to \infty} \braket{ C_B^{(2)} \left( \mathcal V_W \right) } < 1 
\end{align}
by Eq.~\eqref{Eq_return_prob_CGP}. Notice that $H_W$ for $W>0$ is generically non-degenerate
%(except over zero-measure set)
so \autoref{Prop_CGP_meaning} applies. We verify this claim by numerical simulations (see \autoref{Fig_arp}).

The Hamiltonian $H_{W=0}$ is degenerate in the ergodic phase, hence the intertwiner $\mathcal V_{H_{W=0}}$ is not well-defined. Nevertheless, as we show in Appendix \ref{Sec_App_degenerate}, for any choice of eigenbasis of $H_W$ it holds that
\begin{align}
\lim_{L \to \infty}  C_B^{(2)} \left( \mathcal V_{W=0} \right) = 1 \,\;,
\end{align}
hence the average coherence  $\braket{ C_B^{(2)} \left( \mathcal V_W \right) }$ unambiguously distinguishes the two behaviors.

The role of the quantity $C_B^{(2)}(\mathcal V_W)$ might seem special as a probe to localization due to its interpretation as average escape probability. In fact, other measures, arising from an information-theoretic viewpoint of coherence have analogous properties. Let's now consider the relative entropy CGP of the intertwiner, namely $C^{(\rel)}_{B}(\mathcal V_W)$. Its value as a function of the system size $L$ for different values of the disorder strength $W$ is plotted in \autoref{Fig_arp}.
%\textbf{Please check the figures. Fig 2 doesn't show what it's meant here}.
In the ergodic phase $W=0$ it diverges logarithmically 
\begin{align}
C^{(\rel)}_{B}(\mathcal V_{W=0}) \sim \log (L) \,\;. \label{Eq_logarithmic_divergence}
\end{align}
This can be easily verified analytically for an intertwiner connecting two mutually unbiased bases, i.e., for $ \left| \braket{i | \phi_j} \right| = 1/\sqrt{L}$ for all $i,j$. In that case Eq.~\eqref{Eq_logarithmic_divergence} holds with equality, as it directly follows from \autoref{Prop_CGP_formulas}. In Appendix \ref{Sec_App_degenerate} we show that the result again holds in the thermodynamic limit independently of the specific choice for the intertwiner.

We now provide a non-rigorous argument to relate the averages $\braket{ C^{(2)}_{B}(\mathcal V_{W>0})  }$ and $\braket{ C^{(\rel)}_{B}(\mathcal V_{W>0})  }$ to the corresponding localization lengths $\xi_j$. In the localized phase, the eigenvectors typically decay exponentially, i.e.,
\begin{align}
\left| \braket{i|\phi_j} \right|^2 \le c_j \exp \left( -|i - \alpha_j| / \xi_j  \right) \,\;,
\end{align}
where $\alpha_j$ is the site around which $\ket{\phi_j}$ is localized, while due to the periodic boundary conditions $\left|  i - \alpha_j \right|$ above should be understood as $\min (\left|  i - \alpha_j \right|,\left|  i - \alpha_j \pm L \right|)$). If one uses the ansatz
\begin{align}
\left( X_{\mathcal V_{W}} \right)_{ji} = \left| \braket{i|\phi_j} \right|^2 = c_j \exp \left( -|i - \alpha_j| / \xi_j  \right) \,\;,
\end{align} 
then for $L \gg 1$
\begin{subequations} \label{Eq_Heuristic}
\begin{align} \label{Eq_Heuristic_1}
    \braket{ C^{(2)}_{B}(\mathcal V_{W>0})  } \cong 1 -\frac{1}{L} \sum_j \frac{\tanh^2[(2\xi_j)^{-1}]}{\tanh(\xi^{-1}_j)} 
\end{align}
and
\begin{multline}  \label{Eq_Heuristic_2}
\braket{ C^{(\rel)}_{B}(\mathcal V_{W>0})  } \cong \\
\frac{1}{L} \sum_{j = 1}^L \left( \left[ \xi_j \sinh (1/\xi_j)  \right]^{-1} - \ln  \left( \tanh\left[ (2 \xi_j)^{-1} \right] \right)  \right) 
\end{multline}
\end{subequations}
(entropy here has natural logarithm). A detailed derivation can be found in Appendix \ref{Sec_App_derivation_tanh}.

%\begin{figure}[t]
%\centering
%\includegraphics[width=0.45\textwidth]{C_entropy.pdf}
%\caption{Log-linear plot of $\braket{ C_B^{(\rel)} \left( \mathcal V_W \right) }$ as a function of the system size $L$ for different values of the disorder strength $W$. The system is in the localized phase for all $W>0$, in which the asymptotic value is finite. In the ergodic phase ($W = 0$) $\braket{ C_B^{(\rel)} \left( \mathcal V_W \right) }$ diverges logarithmically. The number of realizations range from \num{30000} for small sizes to just 8 for the largest size. Error bars represent standard deviations. Entropy has logarithm with base 2.} \label{Fig_arp}
%\end{figure}

The expression \eqref{Eq_Heuristic_2} for $\xi_j \gg 1$ can be expanded as $\braket{C_B^{(\rel)}\left( \mathcal V_{\Gamma}  \right) } = \frac{1}{L} \sum_{j=1}^L \left( 1 + \ln \left( 2 \xi_j \right) + O(\xi_j^{-2}) \right)$, which is consistent with the numerically observed behavior that it remains finite in the localized phase while it diverges logarithmically as a function of $L$ in the ergodic one.
%\textbf{Can you explain a bit better, i.e.~diverges with what?}.

The accuracy of equations \eqref{Eq_Heuristic} can be assessed by comparing with cases for which an analytical expression can be obtained for the localization lengths $\xi_j$ as a function of the disorder strength. We now consider such a case, described by a Hamiltonian as in Eq.~\eqref{Eq_Anderson_Hamiltonian}, but with on-site energies that follow a Cauchy distribution with parameter $\Gamma$ and vanishing mean (also known as Lloyd model \cite{Lloyd_1969}). We focus for concreteness on Eq.~\eqref{Eq_Heuristic_1} and we denote the corresponding Hamiltonian and intertwiner as $H_\Gamma$ and $\mathcal V_\Gamma$, respectively.
%Starting from Eq.~\eqref{Eq_Heuristic_1}
Utilizing a well-known result from Thouless \cite{Thouless_1972} that connects the localization length with the energy spectrum, one can express the RHS of Eq.~\eqref{Eq_Heuristic_1} as a function of the disorder strength $\Gamma$. This allows for a direct comparison with numerical evaluations of the mean $\braket{ C_B^{(2)} ( \mathcal V_\Gamma) }$, yielding a sound agreement for small disorder ($\Gamma < 1$). We present the details in \autoref{Sec_Cauchy} of the Appendix.

%since the lengths $\xi_j$ are finite \geo{with probability 1} in the localized phase while they diverge in the ergodic one (in the thermodynamic limit), the 

%why in the localized phase the thermodynamic limit
%\begin{align}
%\lim_{L \to \infty} \braket{ C^{(\rel)}_{B}(\mathcal V_{W>0})  }_\epsilon < \infty
%\end{align}
%should be finite for any $W>0$, in contrast with the ergodic $W=0$ case. 

\section{Coherence-generating power and many-body localization} \label{Sec_MBL}

%We now consider disordered quantum many-body system admitting a phase diagram with an ergodic phase (at low enough disorder) and an MBL phase at strong disorder. As in the 1-D Anderson model, we want to characterize the behavior $\braket{ C_B^{(2)} \left( \mathcal V \right) }$ and $\braket{ C_B^{(\rel)} \left( \mathcal V \right) }$ as a function of the disorder strength, where $V$ is the intertwiner between the relevant localized basis and the Hamiltonian eigenbasis and $\braket{\cdot}$ denotes averaging over disorder.

%We will now proceed, in a non-rigorous fashion, to consider what happens in the thermodynamic limit ``$L\to\infty$'', where we denote $d = d^L$ \geo{this is old notation, fix}. 

%In the  MBL phase we \textit{conjecture} that $\braket{ C_B^{(2)} \left( \mathcal V \right) }$ in the thermodynamic limit converges to a finite value, bounded away from 1. \geo{This should be known in the survival probability context. I still need to study the Refs. that I shared with you guys and ``import'' the stuff that is known. Maybe we can say even more things than asymptotic behavior.}

%\geo{Here we should write another part describing what we see numerically for the $\braket{ C_B^{(\rel)} \left( \mathcal V \right) }$.}

We now turn to a disordered quantum many-body system admitting a phase diagram with an ergodic phase at low enough disorder and an MBL phase at strong disorder. For this purpose, we consider a transverse-field Heisenberg spin-1/2 chain in a random magnetic field (along the $\hat{z}$-axis) over $L$ sites ($d = 2^L$) with periodic boundary conditions, described by the Hamiltonian

%As in the Anderson model, we want to study the behavior of $\braket{ C_B^{(2)} \left( \mathcal V_W \right) }$ and $\braket{ C_B^{(\rel)} \left( \mathcal V_W \right) }$ as a function of the disorder strength, where $V_W$ is the intertwiner between the relevant localized basis and the Hamiltonian eigenbasis and $\braket{\cdot}$ denotes averaging over disorder.

%\NA{Some details about the numerics for MBL}\\
%MBL in the presence of strong disorder is a mechanism for quantum systems to escape thermalization and as such, a complete description of these systems is not encapsulated by conventional statistical mechanics.

%More precisely, we study numerically the localization transition in a random-field XXZ chain, described by the Hamiltonian

\begin{align}
H _ { \mathrm { XXX } } &= \frac { 1 } { 2 } \sum _ { i = 1 } ^ { L  } \left[ \sigma _ { i } ^ { x } \sigma _ { i + 1 } ^ { x } + \sigma _ { i } ^ { y } \sigma _ { i + 1 } ^ { y } +  \sigma _ { i } ^ { z } \sigma _ { i + 1 } ^ { z } \right] \nonumber \\
&+ h_x \sum _ { i = 1 } ^ { L } \sigma _ { i } ^ { x } +  \sum _ { i = 1 } ^ { L } w _ { i } \sigma _ { i } ^ { z } \,\;,
\end{align}
where the $h_x$ is the strength of the transverse field and the local field strengths are i.i.d.~random variables with uniform distribution $w_{i} \in \left[ -W,W \right]$. Notice that the transverse field breaks the rotational symmetry of the Hamiltonian. The model has been extensively studied numerically and is known to exhibit a transition from an ergodic to an MBL phase at disorder strength  $W_C \approx 3.7$ (in the absence of the transverse field term), see \cite{Luitz_Lev_review_2017,RevModPhys.91.021001} and references therein.

\begin{figure}[t]
\centering
\includegraphics[width=0.47\textwidth]{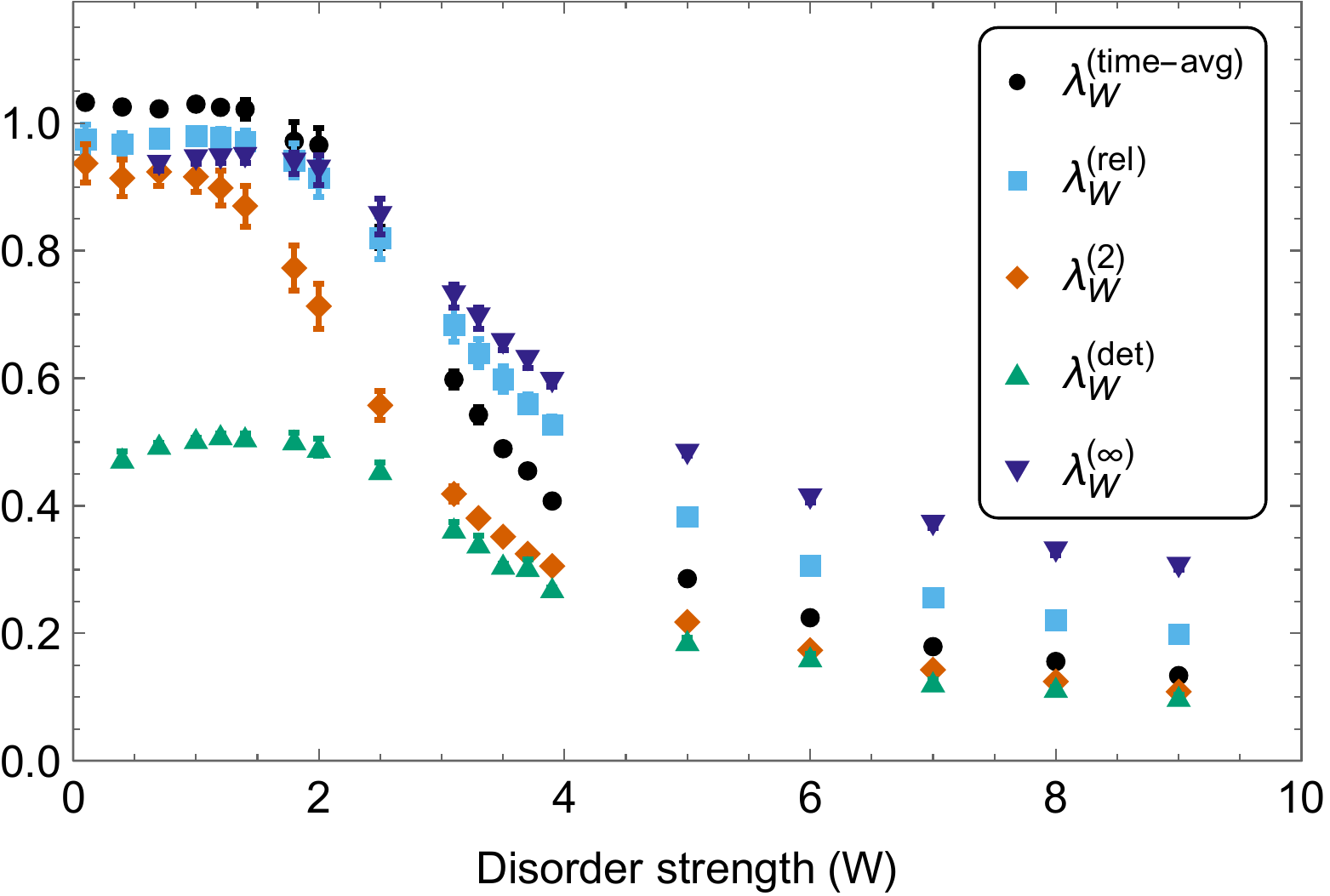}
\caption{Asymptotic behavior for the slope of the quantities: $\log_2 \left( 1 - \braket{ C_B^{(2)} \left( \mathcal V_W \right) } \right) = \log_2 \left( P_{\text{return}}  \right)$, $\braket{ C_B^{(\text{rel})} \left( \mathcal V_W \right) }$, $\log_2 \left( 1 - \braket{ f_B^{(\det)} \left( X_{ \mathcal V_W } \right)} \right)$, $\log_2 \left( 1 - \braket{ f_B^{(\text{time-avg})}\left( X_{ \mathcal V_W } \right)} \right)$, and $\log_2 \left( 1 - \braket{ f_B^{(\infty)}\left( X_{ \mathcal V_W } \right)} \right)$  for large $L$ as a function of the disorder strength $W$ for the Hamiltonian $H_{\mathrm{XXX}}$ at $h_x=0.3$. The slope was extracted for sizes $L = 4, 6, \cdots, 14$, with sample sizes $\num{20000}, \num{20000}, \num{20000}, \num{8000}, \num{2000}, \num{800}$; except at the $W=3.7$, where the sample sizes were doubled. The error bars represent the standard error of the linear fit (see the Appendix~\ref{App_numerics} for more details). Entropy has logarithm with base 2.}
\label{Fig_MBL}
\end{figure}

Similar to the Anderson Hamiltonian, we first study the behavior of the CGP $\braket{ C_B^{(2)} \left( \mathcal V_W \right) }$ and $\braket{ C_B^{(\rel)} \left( \mathcal V_W \right) }$, where $\mathcal V_W$ is the intertwiner between the Hamiltonian and the configuration space basis, which here is taken to be the product $\bigotimes_i \sigma_i^z$ eigenbasis. We find a distinct behavior of the quantities $\braket{ C_B^{(2)} \left( \mathcal V_W \right) }$ and $\braket{ C_B^{(\rel)} \left( \mathcal V_W \right) }$ between the ergodic and MBL phases of the model, as also hinted from the numerical results in \cite{PhysRevLett.110.260601,De_Luca_2013,PhysRevB.92.180202,PhysRevB.92.014208,TH_Santos_Extended_2017,PhysRevB.96.014202}.

For sizes up to $L = 14$, none of the studied CGP quantities seems to reach a constant asymptotic value as in the Anderson case. Nonetheless,
%we numerically study
the (average) return probability $P_{\text{return}}$
%\begin{align}
%    P_{\text{return}} \coloneqq 1 - \mathcal P_{B'} = 1 - \braket{ C_B^{(2)} \left( \mathcal V_W \right)}
%\end{align}
as a function of the number of spins $L$ is consistent with an exponential decay 
\begin{align} \label{Eq_rate_2}
    P_{\text{return}} \propto  2 ^{- \lambda^{(2)}_W L} = d^{- \lambda^{(2)}_W} \,\;.
\end{align}
The extrapolated rates $\lambda^{(2)}_W$, plotted in \autoref{Fig_MBL}, are close to 1 in the ergodic phase, while they drop at the transition point, obtaining a significantly reduced value at the MBL phase. On the other hand, the relative entropy CGP is consistent with a scaling
\begin{align}
    \braket{ C_B^{(\rel)} \left( \mathcal V_W \right) } =  \lambda^{(\rel)} _W L + \text{const}  \,\;,
\end{align}
%(see \autoref{Fig_MBL}),
with a rate $\lambda^{(\rel)} _W$ that is close to 1 for small disorder and drops significantly in the MBL phase.

We now turn to the generalized CGP measures $f_B^{(\det)}$ and $f_B^{(\infty)}$, whose behavior is also consistent with a scaling
\begin{align} \label{Eq_rate_det}
    1 - \braket{ f_B^{(x)} \left( X_{ \mathcal V_W } \right)} \propto  2 ^{- \lambda^{(x)}_W L} = d^{- \lambda^{(x)}_W} \,\;,
\end{align}
and a rate $\lambda_W^{(x)}$ showing distinct behavior in the different phases  ($x = \det$ or $x = \infty$). In \autoref{Sec_App_comparison_lambdas} of the Appendix we show that
\begin{align} \label{Eq_comparison_lambdas}
    \lambda_W^{(\det)} \ge \frac{1}{2} \lambda_W^{(2)} \,\;,
\end{align}
which is saturated for small disorder values and is verified by the observed numerical simulations. Exponential decay is also encountered for the time-average $\braket{ f_B^{(\text{time-avg})}\left( X_{ \mathcal V_W } \right)}$, also plotted in \autoref{Fig_MBL}. Notice that, although the latter fails to be a generalized CGP measure, it can still be employed to detect the transition. For more details about the numerical simulations see the Appendix~\ref{App_numerics}.

For what regards $\braket{ C_B^{(2)} \left( \mathcal V_W \right) }$, we can obtain its behaviour in the limit of infinite disorder and in the ergodic phase. First, we write the  \textit{return probability} $P_{\text{return}} \coloneqq 1 - \mathcal P_{B'} = 1 - \braket{ C_B^{(2)} \left( \mathcal V_W \right)}$ as
\begin{align}
    P_{\text{return}} = \frac{1}{d}\sum_{i=1}^d \langle i | \mathcal{E} ( | i\rangle \langle i | ) |i \rangle \,\;,
\end{align}
where $\mathcal{E} \coloneqq \langle \overline{{\cal U}_t} \rangle$ is the average of the quantum (superoperator) evolution  
$\mathcal{U}_t (\cdot) = e^{-itH_{\mathrm{XXX}}} (\cdot) e^{itH_{\mathrm{XXX}}}$
and $ |i\rangle$ denotes the product $\bigotimes_i \sigma_i^z$ (Ising) basis. 

In the limit of strong disorder $\mathcal{E} ( | i\rangle \langle i | )= | i\rangle \langle i | $ so that $P_{\text{return}} =1$. 
Instead, in the ergodic phase, or more precisely assuming that the operators $| i\rangle \langle i |$ are shell-ergodic (see \cite{LCV_Liu2019}) one obtains $\mathcal{E} ( | i\rangle \langle i | )= \rho_{eq} $ where $\rho_{eq}$ is the (microcanonical) equilibrium state (see \cite{LCV_Liu2019} for more details). This implies that $P_{\text{return}} = (1/d) \mathrm{Tr} (\rho_{eq} ) = 1/d$ thus converging to zero in the thermodynamic limit. One reaches the same conclusion ($P_{\text{return}} \to 0$ albeit possibly with a different speed) if shell ergodicity holds not for all but for sufficiently many basis projectors $|i\rangle \langle i|$.

Finally, we comment on our findings from the typicality point of view. In \cite{coherence_1} it was shown that, if the intertwiner is chosen at random from the unitary group $V \in U(d)$ according to the Haar measure, then $C^{(2)}_B$ is concentrated near its mean
\begin{align}
\braket{C^{(2)}_B (\mathcal V)} _{\text{Haar}}= 1 - \frac{2}{d+1}
\end{align}
($\braket{\cdot}_{\text{Haar}}$ denotes the Haar average over the intertwiner), with overwhelming probability for large Hilbert space dimension $d$ (here $B$ can be any fixed basis). In other words, the typical rate for $P_{\text{return}}$ is $\lambda^{(2)}_{\text{Haar}} \approx 1$. From that perspective, an ergodic behavior is the typical one, while the MBL case can be seen as a highly atypical outlier.

%, which involved various quantifiers of the ``spread'' of the Hamiltonian eigenstates with respect to the configuration space basis. 

%$$ H _ { \mathrm { XXZ } } = \frac { 1 } { 2 } \sum _ { i = 1 } ^ { L - 1 } \left[ \sigma _ { i } ^ { x } \sigma _ { i + 1 } ^ { x } + \sigma _ { i } ^ { y } \sigma _ { i + 1 } ^ { y } + J _ { z } \sigma _ { i } ^ { z } \sigma _ { i + 1 } ^ { z } \right] + \sum _ { i = 1 } ^ { L } w _ { i } \sigma _ { i } ^ { z }. $$

%The first part of the Hamiltonian is the XXZ model and the second term is the random magnetic field.

%This model can be in the ergodic or the MBL phase depending on the value of the interaction strength $J_{z}$ and the disorder strength $W$ which controls the width of the uniform distribution of random fields, $w_{i} \in \left[ -W,W \right]$.

%\NA{cite Abanin RMP 2019}
%- For $J_{z} = 1$, the model is in the many-body localized phase for $W_{c} \geq 3.7$ even at infinite temperature. This means that all many-body eigenstates even in the middle of the band are localized.

%- When $J_{z} = 0$, the model maps onto an Anderson insulator of free fermions and is always localized.

%These behaviors can be interpreted by saying that in the ergodic phase a perturbation of the Hamiltonian results in a new eigenstate system
%that is as far as possible from the unperturbed one. Whereas, in the MBL phase the eigenstate system of the  perturbed Hamiltonian is still close to the original one.

%\geo{Say about typicality results and CGP. From this point of view the MBL case can be seen as an highly atypical outlier.}

\section{Differential geometry of coherence-generating power and MBL} \label{sec_diffgeom}

In this section we study the behavior of the CGP $C^{(2)}_{B} (\delta \mathcal V)$ when the intertwiner $\delta \mathcal  V$ connects two bases that are ``infinitesimally close'' to each other. This results in a differential-geometric construction whose central quantity is a Riemannian metric. As we will show, the resulting metric
\begin{inparaenum}[(i)]
\item is directly connected to the dynamical conductivity, which is a quantity of experimental relevance, and
\item behaves distinctly in the MBL and ergodic phases.
\end{inparaenum}
The detailed mathematical structure is presented in Appendix \ref{sec_diffgeomsupp}.
%the basis $B = \left\{ P_i \right\}_i$ with another basis $B' = \left\{ \mathcal V \left( P_i \right) \right\}_i$ that is infinitesimally close to $B$.

%\subsection{CGP and distance in the Grassmannian}

Consider a complete orthonormal family of states $\left\{\ket{\phi_i(\lambda)}\right\}_{i=1}^{d}$, parametrized by a set of parameters $\lambda$. This is the relevant case, for instance, when one studies the eigenvectors associated with a family of Hamiltonians $H(\lambda)$. The infinitesimal adiabatic intertwiner $\delta{\mathcal V}$
%\geo{I took out the $ad$ subscript, maybe I'm missing its meaning though}
is a unitary map defined by
\begin{align}
\delta{\mathcal V} \left( \ket{\phi_i(\lambda)}\! \bra{\phi_i(\lambda)} \right)
%  \delta{\mathcal V}^\dagger
=\ket{\phi_i(\lambda+ d\lambda)} \! \bra{\phi_i(\lambda+ d\lambda)} \,\;,
\end{align}
where $H(\lambda)\ket{\phi_i (\lambda)} = E_i(\lambda)\ket{\phi_i(\lambda)}$.

%The latter, in view of Eq.~\eqref{C_B-D}, has in turn the physical interpretation as the $C_B^{(2)}$ of the unitary associated with an infinitesimal transformation $\left\{\ket{\phi_i(\lambda)}\right\}_{i=1}^{d} \mapsto \left\{\ket{\phi_i(\lambda+ d\lambda)}\right\}_{i=1}^{d}$. This is the relevant case when one has a family
%of Hamiltonians $H(\lambda)$ (where $\lambda$ denotes a set of control parameters) and wants to compare the eigenstate systems of $H(\lambda)$ and $H(\lambda + d\lambda)$. 

It can be shown that the CGP of
%the infinitesimal adiabatic intertwiner 
$\delta{\mathcal V}$ has the form $C^{(2)}_B(\delta{\mathcal V}) = 2 g d\lambda^2$, where $g$ is a metric given by
\begin{subequations} \label{Eq_metric}
\begin{align}
g & \coloneqq \frac{1}{d}
\sum_{i=1}^{d} \chi_i \,\;, \\
\chi_i &\coloneqq  \Braket{\frac{\partial \phi_i}{\partial \lambda}| \frac{\partial \phi_i}{\partial \lambda}} - \Braket{\phi_i| \frac{\partial \phi_i}{\partial \lambda}} \Braket{ \frac{\partial \phi_i}{\phi_i} | \phi_i } \,\;,
\label{diff-dist}
\end{align}
\end{subequations}
i.e., it is itself a mean of the metrics $\chi_i$ which are associated to the vectors $\ket{\phi_i}$. When the latter are Hamiltonian eigenstates, $\chi_i$ are known as \textit{fidelity susceptibilities} \cite{PhysRevE.74.031123,PhysRevLett.99.095701,PhysRevE.76.022101} and the ground state susceptibility $\chi_0$  plays a key role in the differential geometric approach to quantum phase transitions \cite{PhysRevLett.99.100603}.
%Notice that the metric \eqref{Eq_metric} arising here is the mean of projective space ones.

%Physically, the ground state susceptibility $\chi_0$  plays a key role in the differential geometric approach to quantum phase transitions (see \cite{PhysRevE.74.031123,PhysRevLett.99.100603,PhysRevLett.99.095701}). The main idea is that when $\chi_0$, which depends of the parameters defining the Hamiltonian, shows some singularity in the thermodynamic limit or a super
%extensive (for local Hamiltonians) behavior for finite-size systems, a quantum phase transition is occurring at that point in the parameter space.

In order to connect with quantities of experimental relevance, let us now consider the thermal analog of the metric $g$. We denote $g_T=\sum_{i} p_{i}\chi_{i} $, where $p_{i}=\exp\left(-E_{i}/T\right)/Z$ are the thermal weights and $Z$ denotes the partition function. The quantity $g_T$, defined in \cite{KolodrubetzGeometry2017} as a
generalization of the fidelity susceptibility at finite temperature ($g = g_{T=\infty}$), can be thought of as the metric associated with the thermal analogue of the CGP $C(\mathcal V, c_B^{(2)},\mu_{T})$, where the measure $\mu_T$ weights the Hamiltonian eigenstates with the associated Gibbs weights. The quantity $g_T$ can be expressed via the (imaginary part of the) dynamical susceptibility $\chi_{VV}(\omega)$, where $V=\partial_{\lambda}H(\lambda)$. More precisely (see \cite{KolodrubetzGeometry2017}),
\begin{align}
g_T=\int_{0}^{\infty}\frac{d\omega}{\pi}\frac{\chi''_{VV}(\omega)}{\omega^{2}}\coth\left(\frac{\omega}{2T}\right).\label{eq:gT2}
\end{align}
%The dynamical susceptibility can usually be obtained in experiments \geo{vague, citation missing}.
The above formula is remarkable, as it demonstrates that the, apparently abstract, quantity $ C_B^{(2)}(\delta  \mathcal V)$ is simply connected with a quantity measurable in experimental setups \cite{PhysRevLett.104.147201,RevModPhys.87.855,hauke2016measuring}. We also note that, although Eq.~\eqref{eq:gT2} is not straightforwardly applicable in the infinite temperature limit, in this limit one obtains
\begin{equation}
g = g_{T=\infty}=\frac{1}{\pi}\int_{-\infty}^{\infty}\frac{\sigma_{VV}(\omega)}{\omega^{2}}d\omega \,\;,
\label{eq:gDC}
\end{equation}
where $\sigma_{VV}(\omega)$ is the high-temperature dynamical conductivity \footnote{The name dynamical conductivity comes from its use when $\chi_{VV}$
is the (charge) current-current correlation.} given by

\begin{align}
    \sigma_{VV}(\omega) & =\frac{2\pi}{d}\sum_{n\neq m}\left|V_{n,m}\right|^{2}\delta[\omega-(E_{m}-E_{n})].
\end{align}

In this case, the role of $g$ is played by the d.c.~dielectric polarizability \cite{PhysRevB.94.045126,Prelovsek2017MBL}.
%In our case $V=\partial_{\lambda}H(\lambda)$ is not necessarily the charge current operator, although in experiments it will likely be. \geo{Ref missing, also what experiments, and how likely?}

% \begin{figure}[t]
% \centering
% \includegraphics[width=0.45\textwidth]{disorder-vs-exponent-relent.pdf}
% \caption{Asymptotic behavior of the slope of $ \braket{ C_B^{(\rel)} \left( \mathcal V_W \right) } $ for large $L$ as a function of the disorder strength $W$. The slope \geo{fractal dimension?} was extracted for sizes up to $L = 14$, with sample sizes \geo{say details}.} \label{Fig_MBL_rel_ent}
% \end{figure}  

The quantities $g_T$ and $g$ not only allow to make contact with experiments, but have also been studied in the context of thermalization and MBL. In particular, it is believed that $g\to\infty$ in the thermodynamic limit, both for the ergodic and the sub-diffusive phase. Instead, in the MBL phase $g\to\mathrm{constant<\infty}$ \cite{Prelovsek2017MBL}.
In the light of Eq.~\eqref{Eq_metric}, these results mean that the CGP of the adiabatic intertwiner between nearby Hamiltonians has distinctively different behaviors in the ergodic and in the MBL phases.

\section{Conclusions and future works} \label{sec_conclusions}

In this work we have brought together ideas from quantum information and geometry, on the one hand, and the physics of disordered systems on the other. We established a connection between the quantitative approach to coherence, originating from the perspective of quantum resource theories \cite{PhysRevLett.113.140401,RevModPhys.89.041003}, and localization \cite{anderson_absence_1958,basko_metal-insulator_2006,pal_many-body_2010,nandkishore_many-body_2015}.

More specifically, we studied the behavior over the ergodic, Anderson and many-body localized phases in terms of the scaling  properties of coherence averages that are associated to the intertwiner connecting the Hamiltonian eigenvectors with the configuration space basis.
%We conjectured  that in the ergodic (MBL) phase in the thermodynamical limit the CGP is maximal (sub-maximal) and provided heuristic theoretical arguments as well as numerical evidence in support of this conjecture.
Furthermore, we built an associated differential-geometric version for infinitesimal perturbations of the Hamiltonian, and showed that the resulting Riemannian metric can be mapped onto known physical quantities which have a sharply distinct  behavior in the ergodic and in the MBL phases.
 
%\geo{fix} In this work, the focus was on two specific measures of coherence, but the framework laid is general and allows for an analogous construction to be carried over for any coherence measure. In fact, the two CGP quantities studied are proper monotones in the framework of incompatibility~\cite{styliaris2019quantifying}, where the latter is cast (for the case of orthogonal measurements) as a resource theory over the set of bases. It remains an open question to characterize the essential properties of coherence or incompatibility quantifiers that may allow for an insightful connection with localization.

Quantum chaos is often dubbed as the \textit{dynamical counterpart} of quantum localization and connections between the two have been used to elucidate the physics of chaotic systems~\cite{chaos_1,haake2010quantum}. Following this correspondence, we conjecture that the CGP can act as a signature of quantum chaos, for example, by identifying the so-called ``edge of chaos''~\cite{weinstein2002edgeofchaos}. Moreover, dynamical quantities like the survival probability~\cite{torres2019survival} and entangling power~\cite{zanardi2000entangling} (which is the direct analogue of CGP for entanglement) have been applied to the study of chaotic systems like the quantum kicked top~\cite{haake2010quantum}, which can be related to measures of CGP, as will be explored in a forthcoming article~\cite{preprint2019namit}.

Establishing further, more rigorous, theoretical as well as numerical grounds for this connection between the information-theoretic approach to quantum coherence and many-body theory provides a direction for future research. 

\begin{acknowledgments}

G.S.~acknowledges financial support from a University of Southern California ``Myronis'' fellowship. N.A.~acknowledges the HPC staff at USC for their assistance. L.C.V.~acknowledges partial support from the Air Force Research Laboratory award no.~FA8750-18-1-0041. P.Z.~acknowledges partial support from the NSF award PHY-1819189.
%This work was partially supported by the Air Force Research Laboratory award no.~FA8750-18-1-0041 and (partially) by the Office of the Director of National Intelligence (ODNI), Intelligence Advanced Research Projects Activity (IARPA), via the U.S. Army Research Office contract W911NF-17-C-0050. The views and conclusions contained herein are those of the authors and should not be interpreted as necessarily representing the official policies or endorsements, either expressed or implied, of the ODNI, IARPA, or the U.S.~Government. The U.S.~Government is authorized to reproduce and distribute reprints for Governmental purposes notwithstanding any copyright annotation thereon. 
The research is based upon work (partially) supported by the Office of the Director of National Intelligence (ODNI), Intelligence Advanced Research Projects Activity (IARPA), via the U.S. Army Research Office contract W911NF-17-C-0050. The views and conclusions contained herein are those of the authors and should not be interpreted as necessarily representing the official policies or endorsements, either expressed or implied, of the ODNI, IARPA, or the U.S. Government. The U.S. Government is authorized to reproduce and distribute reprints for Governmental purposes notwithstanding any copyright annotation thereon.

\end{acknowledgments}

\appendix

\section{Proofs of Propositions} \label{Sec_app_proofs}

\cgpformulas*

\begin{proof}
\textbf{(i)} We follow a procedure similar to the one in Ref.~\cite{coherence_1}. We make use of the Hilbert-Schmidt inner product $\braket{A, B} \coloneqq \Tr \left( A^\dagger B \right)$ over the space $\mathcal B(\mathcal H)$ of bounded linear operators over $\mathcal H$. Starting from Eq.~\eqref{Eq_CGP_general} with $c_B = c_B^{(2)}$, we get
\begin{align*}
C^{(2)}_B \left( \mathcal U \right) & = \frac{1}{d} \sum_i \left\| \left( \mathcal I - \mathcal D_B \right) \mathcal U \Pi_i  \right\|_2^2 \\
& =  \frac{1}{d} \sum_i  \braket{\left( \mathcal I - \mathcal D_B \right) \mathcal U \Pi_i,\left( \mathcal I - \mathcal D_B \right) \mathcal U \Pi_i}  \\
& =  \frac{1}{d} \sum_i \left( \left\| \mathcal U \Pi_i  \right\|_2^2 - \left\|  \mathcal D_B \mathcal U \Pi_i  \right\|_2^2 \right) \,\;,
\end{align*}
where we have used the fact that the dephasing superoperator $\mathcal D_B \in \mathcal B (\mathcal B (\mathcal H))$ is self-adjoint $\mathcal D_B^\dagger = \mathcal D_B$ with respect to the Hilbert-Schmidt inner product, as well as a projection $\mathcal D_B^2 = \mathcal D_B$. Unitary invariance of the 2-norm implies $\left\| \mathcal U \Pi_i  \right\|_2^2 = 1$. Using the definition Eq.~\eqref{Eq_dephasing_def}, a straightforward calculation gives
\begin{align} \label{Eq_C2B_alt}
C^{(2)}_B \left( \mathcal U \right) = 1 -  \frac{1}{d} \sum_{ij} \left(X_{\mathcal U} \right)^2 _{ji}
\end{align}
which reduces to the claimed result.
\\ \\
\textbf{(ii)} Let us denote the Shannon entropy of a probability vector as $H(\bs p) \coloneqq - \sum_i p_i \log (p_i) $. Since $S(\mathcal U \Pi_i) = S(\Pi_i) =  0$, Eq.~\eqref{Eq_CGP_general} 
with $c_B = c_B^{(\rel)}$ gives
\begin{align*}
C^{(\rel)}_B \left( \mathcal U \right) & = \frac{1}{d} \sum_i S(\mathcal D_B \mathcal U \Pi_i) \\
 & =  \frac{1}{d} \sum_i S \left( \sum_j \left( X_{\mathcal U} \right)_{ji} \Pi_j \right) \\
  & = \frac{1}{d} \sum_i H \left(\left\{ \left( X_{\mathcal U} \right)_{ji} \right\}_j \right) \\
  & = H(X_{\mathcal U}) \,\;.
\end{align*}
\end{proof}

\cgpproperties*

\begin{proof}

We first show that, for a fixed coherence measure $c_B$, the quantity $C \left( \mathcal U , c_B ,\mu _{\unif} \right)$ (explicitly given in Eq.~\eqref{Eq_CGP_general}) can be expressed as a function of $X_{\mathcal U}$. This implies that the phases of $U$ (considered as a matrix in the $B = \{ \Pi_i \}_i = \{ \ket{\phi_i} \!\bra{\phi_i} \}_i$ basis, where $\mathcal U(X) = U X U^\dagger$) are irrelevant.

Consider a pure state $\ket{\psi}$. The value of $c_B\left(  \ket{\psi} \! \bra{\psi} \right)$ can only depend on the modulus of the coefficients $\left\{ \left| \braket{\phi_i | \psi } \right|  \right\}_{i=1}^d$. This follows from the fact that the unitary transformations $\mathcal V (\rho)= V \rho V^\dagger $, such that $V \ket{\psi}$ alters the phases or permutes the coefficients $\{  \braket{\phi_i | \psi }   \}_{i=1}^d$, form a subgroup of the Incoherent Operations. Hence all coherence monotones should maintain a constant value over a group orbit.
%As a result, all phases associated with the coefficients $\left\{ \left| \braket{i | \psi } \right|  \right\}_{i=1}^d$ are irrelevant as far as coherence is concerned.
As a result, $c_B(\mathcal U (\Pi_j))$ can be expressed as a function of $\{ \left( X_{\mathcal U} \right)_{ij} \}_{i=1}^d$ (recall $\left( X_{\mathcal U} \right)_{ij} = \left| \braket{\phi_i | U | \phi_j} \right|^2$). Hence, also $C \left( \mathcal U , c_B ,\mu _{\unif} \right)$ can be expressed as a function of the whole matrix $X_{\mathcal U}$ (in fact, an additive one over the columns).

%$C \left( \mathcal U , c_B ,\mu _{\unif} \right)$ does not depend on the phases associated with the pure states $\{ \mathcal U \left( P_i \right) \}_i$, i.e., can be considered as a function of $X_{\mathcal U}$.

Property (i) follows directly from the fact that coherence measures vanish over incoherent states. For property (ii), invariance under pre-processing by a permutation $\Pi'$ holds since the averaging over the states is uniform. Invariance under post-processing by $\Pi$ holds since unitary transformations that permute the elements of $B$ belong to Incoherent Operators.

We now prove property (iii). First notice that, since the value of $c_B\left( \ket{\psi} \! \bra{\psi} \right)$ can only depend on the moduli of the coefficients $\left\{ \left| \braket{\phi_i | \psi } \right|  \right\}_{i=1}^d$, the function $f_B(X)$ is in fact well-defined over all bistochastic matrices (and not just unistochastic \footnote{A bistochastic matrix $M_{ij}$ is called unistochastic if there exists a unitary matrix $U_{ij}$ such that $M_{ij} = \left|  U_{ij} \right|^2 $ (see \cite{bengtsson2017geometry} for more details).} ones).

%To prove the desired inequality of (iii), we will show that there exists a collection of pure states $\{ \ket{\psi_j}\!\bra{\psi_j}  \}_{j=1}^d$ such that $c_B \left( \ket{\psi_j}\!\bra{\psi_j} \right) \ge c_B (\mathcal U(P_j)) $ $\forall j$

Consider a collection of pure states $\{ \ket{\psi_j}\!\bra{\psi_j}  \}_{j=1}^d$ such that
\begin{align}
    \ket{\psi_j} = \sum_i \sqrt{(MX_{\mathcal U})_{ij}} \ket{\phi_i} \,\;.
\end{align}
Then, one has that
%From the discussion at the beginning of the proof, it follows that altering the matrix $X_{\mathcal{U}}$ as $X_{\mathcal{U}} \mapsto M X_{\mathcal{U}}$ (i.e., by post-processing it with $M$) amounts to altering $C \left( \mathcal U , c_B ,\mu _{\unif} \right)$ as
%\begin{align*}
%    \frac{1}{d} \sum_{j=1}^d c_B \left( \mathcal U(P_j) \right)  \mapsto \frac{1}{d} \sum_{j=1}^d c_B \left( \ket{\psi_j}\!\bra{
%    \psi_j} \right)  \,\;,
%\end{align*}
%such that the diagonal part of each $\ket{\psi_j}\!\bra{\psi_j} $ is given by $(M X_{\mathcal U})_{ij}$, i.e.,
\begin{align} \label{Eq_condition_psi}
\Tr \left( \Pi_i \ket{\psi_j} \! \bra{\psi_j} \right) = \sum_k M_{ik}  \Tr \left( \Pi_k \, \mathcal U(\Pi_j) \right) \quad \forall i,j \,\;.
\end{align}
%Let us first explain why one can always find pure states $ \{ \ket{\psi_j} \! \bra{\psi_j} \}_{j=1}^d$ satisfying Eq.~\eqref{Eq_condition_psi}. Indeed, the above condition amounts to finding a pure state whose density matrix has a fixed diagonal part (in the $B$ basis) and its existence is guaranteed, e.g.,  by the Schur-Horn theorem \cite{bhatia2013matrix}.

To prove the desired inequality of (iii), we will show that $c_B \left( \ket{\psi_j}\!\bra{\psi_j} \right) \ge c_B (\mathcal U(\Pi_j)) $ $\forall j$. Indeed, the previous holds true for all coherence measures $c_B$ if for every $j$ there exists an Incoherent Operator $\mathcal E$ such that $\mathcal E \left( \ket{\psi_j}\!\bra{\psi_j}  \right) = \mathcal U(\Pi_j)$. The last is guaranteed (in fact, within Strictly Incoherent Operators) by the main result of \cite{PhysRevA.91.052120} which can be applied since, by the bistochasticity of $M$, Eq.~\eqref{Eq_condition_psi} implies that $ \mathcal D_{B} \left( \mathcal U (\Pi_j) \right) \succ \mathcal D_{B} \left( \ket{\psi_j}\!\bra{\psi_j} \right)$.

\end{proof}

\cgpmajorization*

\begin{proof}
The first part follows by generalizing the proof of part (iii) of \autoref{Prop_CGP_properties}. One can directly extend the construction by considering two sets of pure states $\{ \ket{\psi_j} \! \bra{\psi_j} \}_{j=1}^d$ and $\{ \ket{\psi'_j} \! \bra{\psi'_j} \}_{j=1}^d$ such that
\begin{subequations} \label{eq_app_psi}
\begin{align}
    \ket{\psi_j} &= \sum_i \sqrt{Y_{ij}} \ket{\phi_i} \\
    \ket{\psi'_j} &= \sum_i \sqrt{X_{ij}} \ket{\phi_i} \,\;.
\end{align}
\end{subequations}
Then the convertibility argument $\ket{\psi_j} \! \bra{\psi_j}  \mapsto \ket{\psi'_j} \! \bra{\psi'_j}$ via strictly incoherent operations applies due to the majorization condition, giving the desired result.

For the converse, we will first show that the functions over pure states $c_B(\ket{\psi}\!\bra{\psi}) = \sum_i \phi \left( \Tr \left( \Pi_i \ket{\psi}\!\bra{\psi}  \right) \right)$ are monotones, where $\phi$ is any continuous concave function. Indeed, from the main result of \cite{PhysRevA.91.052120}, a conversion $\ket{\psi} \! \bra{\psi}  \mapsto \ket{\psi'} \! \bra{\psi'}$ via Strictly Incoherent Operations is possible if and only if $\mathcal D_B(\ket{\psi'} \! \bra{\psi'}) \succ \mathcal D_B(\ket{\psi} \! \bra{\psi})$ \footnote{In fact, the majorization condition is only sufficient for convertibility. It becomes also necessary if an additional condition about the rank of the dephased states is satisfied (see \cite{PhysRevA.91.052120} for more details). Nevertheless, if one considers convertibility with some error (arbitrarily small), which is the relevant notion in all physical scenarios, the rank conditions becomes irrelevant.}. However, a standard result by Hardy, Littlewood and P\'{o}lya states that for two probability vectors it holds that $\bs p \succ \bs q$ if and only if $\sum_i \phi(p_i) \le \sum_i \phi(q_i)$ for all continuous concave $\phi$ \cite{marshall_inequalities:_2011}. As a result, $\mathcal D_B(\ket{\psi'} \! \bra{\psi'}) \succ \mathcal D_B(\ket{\psi} \! \bra{\psi})$ is equivalent to $ \sum_i \phi \left( \Tr \left( \Pi_i \ket{\psi'}\!\bra{\psi'}  \right) \right) \le \sum_i \phi \left( \Tr \left( \Pi_i \ket{\psi}\!\bra{\psi}  \right) \right)$, i.e., the aforementioned functions $c_B$ are monotones over pure states.

By assumption, the functions $f_B$ arise from continuous coherence monotones over pure states. From the statement in the previous paragraph it then follows that, in fact, all $f_B(X) = \sum_{ij} \phi \left(X_{ij} \right)$ for continuous concave $\phi$ are such functions. Hence, $ \sum_{ij} \phi \left(X_{ij} \right) \le \sum_{ij} \phi \left(Y_{ij} \right)$. Finally, the aforementioned result by Hardy, Littlewood and P\'{o}lya  \cite{marshall_inequalities:_2011} 
in the context of column majorization implies  $X \succ ^c Y$.

%For the converse, we will invoke the standard result by Hardy, Littlewood and P\'{o}lya  \cite{marshall_inequalities:_2011} which, in the context of column majorization, implies that $X \succ ^c Y$ if and only if $\sum_{ij} \phi\left( X_{ij} \right) \le \sum_{ij} \phi\left( Y_{ij} \right)$ for all continuous concave functions $\phi$. As a result, we need to show that the functions $f_B(X)$ resulting from all coherence monotones $c_B$ include the functions $\{ \sum_{ij} \phi\left( X_{ij} \right) \}_{\phi}$ for all continuous concave $\phi$.  Indeed, by the convertibility result of \cite{PhysRevA.91.052120}, $\ket{\psi_j} \! \bra{\psi_j}  \mapsto \ket{\psi'_j} \! \bra{\psi'_j}$ via strictly incoherent operations if and only if $X^c_j \succ Y^c_j$ $\forall j$ \footnote{In fact,  majorization condition is only sufficient for convertibility. It becomes also necessary if an additional condition about the rank of the dephased states is satisfied (see \cite{PhysRevA.91.052120} for more details). Nevertheless, if one considers convertibility with some error, arbitrarily small, the rank conditions becomes irrelevant.}. Hence, invoking again the result by Hardy, Littlewood and P\'{o}lya, all continuous concave functions obey $\sum_i \phi(X_{ij}) \le \sum_i \phi(Y_{ij})$ $\forall j$, hence the desired result follows. 

\end{proof}

\cgpmeaning*

\begin{proof}
\textbf{(i)} The key observation is that the dephasing superoperator $\mathcal D_B$
arises as the (infinite) time average of the Schr\"odinger evolution $\mathcal{U}_{t}(\cdot)=e^{-itH } (\cdot) e^{itH}$, namely $\overline{\mathcal U_t} = \mathcal D_B$. Using the Hilbert-Schmidt inner product over $\mathcal B(\mathcal H)$ (see proof of \autoref{Prop_CGP_formulas}) and setting $\Pi_\psi = \ket{\psi}\!\bra{\psi}$,  we get
\begin{align*}
\mathcal P_\psi &= 1 - \overline{ \Tr \left( \Pi_\psi \,  \mathcal U_t (\Pi_\psi)\right)} = 1 -  \Tr \left( \Pi_\psi \,  \mathcal D_B (\Pi_\psi)\right) \\
& = 1 - \braket{\Pi_\psi , \mathcal D_B \Pi_\psi} =  1 - \braket{\mathcal D_B\Pi_\psi , \mathcal D_B \Pi_\psi} \\
& = \braket{(\mathcal I - \mathcal D_B) \Pi_\psi , (\mathcal I - \mathcal D_B) \Pi_\psi} \\
& = \left\| (\mathcal I - \mathcal D_B) \Pi_\psi \right\|_2^2 = c_{B}^{(2)} (\Pi_\psi) \,\;.
\end{align*}
\textbf{(ii)} The first equality of Eq.~\eqref{Eq_return_prob_CGP_2} follows by combing part $(i)$ of the Proposition with Eq.~\eqref{Eq_CGP_general}. For the second equality, from the unitary invariance of the 2-norm, we have
\begin{align*}
C_{B}^{(2)}(\mathcal V) &=  \frac{1}{d} \sum_i \left\|(\mathcal I - \mathcal D_B ) \ket{i}\!\bra{i} \right\|_2^2 \\
& =  \frac{1}{d} \sum_i \left\|\mathcal V^\dagger (\mathcal I - \mathcal D_{B'} ) \mathcal V \left( \ket{i} \! \bra{i} \right) \right\|_2^2 \\
& =   \frac{1}{d} \sum_i \left\| (\mathcal I - \mathcal D_{B'} ) \mathcal V \left( \ket{i} \! \bra{i} \right) \right\|_2^2 \\
& = C_{B'}^{(2)}(\mathcal V) \,\;. 
\end{align*}
However, notice that $X_{\mathcal V^\dagger} = X_{\mathcal V}^T$ which from Eq.~\eqref{Eq_2CGP_Tr} implies $C_{B'}^{(2)}(\mathcal V ^\dagger) = C_{B'}^{(2)}(\mathcal V )$.
\end{proof}

\section{Time-averaged CGP} \label{Sec_App_time_avg}

In this section we study the time-average of the CGP $ \overline{C_{B'}^{(2)} (\mathcal U_t)}$, where $\mathcal U_t (X) = \exp(- i H t) X \exp (i H t)$ is the time evolution operator. For the following, we will assume that the Hamiltonian $H = \sum_i E_i \ket{\phi_i} \! \bra{\phi_i}$ satisfies the \textit{non-resonance condition}, i.e., its energy gaps are non-degenerate. Under this assumption, we will show that
\begin{align} \label{Eq_app_CGP_NRC}
    \overline{C_{B'}^{(2)} (\mathcal U_t)} = 
    1 - \frac{2}{d} \sum_{ij} \braket{X^c_i , X^c_j}^2 + \frac{1}{d} \sum_{i} \braket{X^c_i , X^c_i}^2
\end{align}
where $V = \sum_{i} \ket{i} \! \bra{\phi_i}$ is the intertwiner between $B = \{ \Pi_i \coloneqq \ket{\phi_i} \! \bra{\phi_i} \}_i$ and $B' = \{ P_i \coloneqq \ket{i} \! \bra{i} \}_i$.

We have,
\begin{align*}
    \overline{C_{B'}^{(2)} (\mathcal U_t)} & =  1 - \frac{1}{d} \sum_i  \overline{ \braket{\mathcal D_{B'}  \,\mathcal U_t \left( P_i \right) , \mathcal D_{B'} \, \mathcal U_t  \left( P_i \right) } } \\
    & =  1 - \frac{1}{d} \sum_i  \overline{  \braket{ P_i  ,  \mathcal U^\dagger_t \mathcal   D_{B'} \, \mathcal U_t  \left( P_i \right) } } \\
        & =1 - \frac{1}{d} \sum_{ijkk'll'} \Big[ \,  \overline{ \exp[i(E_k - E_{k'} + E_l - E_{l'})] } \\ & \quad \qquad \qquad  \cdot  \Tr \left(  P_i   \Pi_k  P_j   \Pi_{k'}  P_i   \Pi_l   P_j   \Pi_{l'}   \right) \Big] \,\;.
\end{align*}
%where $ \{ \Pi_i \coloneqq \ket{\phi_i} \! \bra{\phi_i} \}_i$ denotes the projectors over the eigenbasis of the Hamiltonian.
The non-resonance condition implies that
\begin{multline*}
     \overline{ \exp[i(E_k - E_{k'} + E_l - E_{l'})] } = \\ = \delta_{kk'} \delta_{ll'} + \delta_{kl'}\delta_{k'l} - \delta_{kk'} \delta_{k'l} \delta_{ll'} \,\;.
\end{multline*}
A straightforward calculation gives
\begin{align*}
    \overline{C_{B'}^{(2)} (\mathcal U_t)} & = 1 - \frac{1}{d} \Big( 2 \sum_{ijkl} (X_{\mathcal V})_{ki} (X_{\mathcal V})_{kj}(X_{\mathcal V})_{li}(X_{\mathcal V})_{lj} \\
    & \qquad \qquad - \sum_{ijk} (X_{\mathcal V})_{ki}^2 (X_{\mathcal V})_{kj}^2 \Big)
\end{align*}
which reduces to Eq.~\eqref{Eq_app_CGP_NRC}.

An easy calculation for a single qubit reveals that $f_B^{(\text{time-avg})} (X)$ is not a generelized CGP measure, since its maximum value is not attained over the transition matrix with elements $X_{ij} = 1/2$ \footnote{In light of the connection between infinite time-average and dephasing, this relates to the results in \cite{coherence_3}, where the question of CGP for dephasing evolutions was investigated.}.

\section{Inverse participation ratio, effective dimension, and Loschmidt echo}
\label{Sec_App_thermalization}

For a non-degenerate Hamiltonian $H = \sum_i E_i \ket{\phi_i} \bra{\phi_i}$, the escape probability $\mathcal P_\psi$ is directly connected with the second Participation Ratio of $\ket{\psi}$ over the Hamiltonian eigenbasis $\PR_2  \coloneqq \sum_i \left| \braket{\phi_i | \psi} \right|^4 $ as $\mathcal P_\psi = 1 - \PR_2$.

The second Participation ratio, in turn, is intimately connected to two other quantities of physical interest in the study of equilibration and thermalization, namely the \textit{effective dimension} and the \textit{Loschmidt echo}~\cite{PhysRevE.79.061103,PhysRevLett.101.190403}. The effective dimension of a quantum state is defined as its inverse purity,
\begin{align}
d^{\mathrm{eff}}(\rho) = \frac{1}{\mathrm{Tr}[\rho^2]},
\end{align}
which intuitively corresponds to the number of pure states that contribute to the (in general) mixed state $\rho$. Given a non-degenerate Hamiltonian, it is easy to show that the effective dimension of the (infinite) time-averaged state is equal to the inverse of the second Participation ratio, that is,
\begin{align}
d^{\mathrm{eff}}(\overline{\rho})=\frac{1}{\operatorname{Tr}\left(\overline{\rho}^{2}\right)}=\frac{1}{\sum_{i}\left|\left\langle\phi_{i} | \psi\right\rangle\right|^{4}} = \frac{1}{\textnormal{PR}_2},
\end{align}
where $\rho = \ket{\psi}\!\bra{\psi}$.

Recall that the Loschmidt echo is defined as the overlap between the initial state $\ket{\psi}$ and the state after time $t$,
\begin{align}
\mathcal{L}_t \coloneqq \left| \braket{\psi | e^{-i H t} | \psi} \right|^2,
\end{align}
the infinite time-average of which can be identified with the \textit{return probability} of the state $\ket{\psi}$. Then, in the non-degenerate case, the time-averaged Loschmidt echo is related to the second Participation ratio and the effective dimension as
\begin{align}
\overline{\mathcal{L}_{t}}= \textnormal{PR}_2 = \frac{1}{d^{\mathrm{eff}}(\overline{\rho})}.
\end{align}

For a more detailed exposition, see \cite{gogolin2016equilibration}.

\section{CGP in the Anderson model for the degenerate case $W = 0$} \label{Sec_App_degenerate}

The spectrum of Anderson Hamiltonian Eq.~\eqref{Eq_Anderson_Hamiltonian} for the disorder-free case is degenerate, hence the intertwiner $\mathcal V _{W = 0}$ between the position and Hamiltonian eigenbases is not uniquely defined. Nevertheless, we show here that the behavior of the quantities $  C^{(2)}_{B}(\mathcal V_{W=0}) $ and $  C^{(\rel)}_{B}(\mathcal V_{W=0}) $ in the thermodynamic limit is independent of the specific choice of the Hamiltonian eigenbasis, namely $  C^{(2)}_{B}(\mathcal V_{W=0}) \to 1$ while $  C^{(\rel)}_{B}(\mathcal V_{W=0}) \sim \log  (L)$ for $L \to \infty$.

The spectrum of the Hamiltonian is $\left\{ 2 \cos \left( \frac{2 \pi j}{L} \right) \right\}_{j=0}^{L-1}$, hence there are $n_L$ distinct two-dimensional degenerate subspaces, where
$n_L = (L-2)/2$ for $L$ even and $n_L = (L-1)/2$ for $L$ odd. Invoking the Fourier eigenbasis 
\begin{align}
\ket{\phi_k} = \frac{1}{\sqrt{L}}\sum_{j=0}^{L-1} \exp \left( - i \frac{2 \pi j k }{L}  \right) \ket{j}  \label{Eq_Fourier_basis}
\end{align}
as reference, the general eigenbasis of $H_{W=0}$ 
may differ from basis \eqref{Eq_Fourier_basis} as
\begin{subequations}
\begin{align}
\ket{\phi'_k} &= e^{i \gamma_k} \left( e^{i \alpha_k} \cos(\theta_k) \ket{\phi_k} + e^{i \beta_k} \sin(\theta_k) \ket{\phi_{L-k}}  \right) \\
\ket{\phi'_{L-k}} &= e^{i \gamma_k} \left( - e^{i \beta_k} \sin(\theta_k) \ket{\phi_k} + e^{- i \alpha_k} \cos(\theta_k) \ket{\phi_{L-k}}  \right) 
\end{align}
\end{subequations}
for $k = 1,\dots,n_L $, where the angles $\left\{ \alpha_k,\beta_k,\gamma_k,\theta_k  \right\}$ specify the (unitary) transformation within the $k^{\text{th}}$ two-fold degenerate subspace.

A straightforward calculation gives
\begin{multline}
\left| \braket{l | \phi'_k}  \right|^2 = \left| \braket{l | \phi'_{L-k}}  \right|^2 = \\ =  \frac{1}{L} \left[ 1 + \cos \left( \frac{2 (L-2k) l \pi}{L} + \alpha_k - \beta_k \right) \sin(2 \theta_k) \right] \,\;. \label{Eq_modified_X}
\end{multline}
from which one can directly see that the possible Hamiltonian eigenbases differ in the sum $\sum_{i,j} \left(X_{\mathcal U} \right)^2 _{ji}$ at most of an order 1 term. Hence, from Eq.~\eqref{Eq_C2B_alt} it follows that any such contribution vanishes at the thermodynamic limit, yielding $  C^{(2)}_{B}(\mathcal V_{W=0})  \to  1$.

For $ C^{(\rel)}_{B}(\mathcal V_{W=0}) $, we first invoke the standard inequality between the Shannon entropy and the purity $H(\left\{ p_i \right\}) \ge  -\log  \left( \sum_i p_i^2 \right)$ (following from the monotonicity of the R\'{e}nyi entropies \cite{cover2012elements}). By the use of Eq.~\eqref{Eq_modified_X}, the purity of the probability distribution $\left\{ \left| \braket{l | \phi'_k}  \right|^2  \right\}_{l=1}^L$ is
\begin{align*}
\sum_{l=1}^L  \left| \braket{l | \phi'_k}  \right|^4 = \frac{2+\sin^2(2 \theta_k)}{2L} \,\;,
\end{align*}
therefore the previous inequality implies
\begin{align*}
H \left( \left\{ \left| \braket{l | \phi'_k}  \right|^2 \right\}_l \right) \ge \log  L - \log  \left( \frac{2+\sin^2(2 \theta_k)}{2}  \right) \,\;.
\end{align*}
Finally, this implies by Eq.~\eqref{Eq_relCGP_H} that $  C^{(\rel)}_{B}(\mathcal V_{W=0})$ diverges logarithmically with $L$ for any choice of the Hamiltonian eigenbasis.

\section{Derivation of Eqs.~\eqref{Eq_Heuristic}} \label{Sec_App_derivation_tanh}

In this section we show how using the ansatz $\left( X_{\mathcal V_{W}} \right)_{ji} = c_j \exp \left( -|i - \alpha_j| / \xi_j  \right)$, one can derive Eqs.~\eqref{Eq_Heuristic}.

Assuming periodic boundary conditions as in the main text, and since $\sum_i \left( X_{\mathcal V_{W}} \right)_{ji} = 1$, the coefficients $c_j$ can be expressed for $L\gg 1$ as
\begin{align*}
(c_j)^{-1} \approx 2 \sum_{x = 0}^{\infty} e^{ -x / \xi_j } - 1
% = 2 \, \frac{ 1 }{1 - e^{1/\xi_j}} - 1 \,\;,
\end{align*}
therefore
\begin{align} \label{Eq_normalization_cj}
c_j = \tanh [ (2 \xi_j)^{-1}  ] \,\;.
\end{align}
From Eq.~\eqref{Eq_2CGP_Tr},
\begin{align*}
    C_B^{(2)} = 1 - \frac{1}{L} \sum_{ij} \left( X_{\mathcal V_{W}} \right)_{ij}^2 = 1 - \frac{1}{L} \sum_j \frac{\tanh^2[(2\xi_j)^{-1}]}{\tanh(\xi^{-1}_j)} \,\;,
\end{align*}
which is \eqref{Eq_Heuristic_1}.

Similarly, from Eq.~\eqref{Eq_relCGP_H} we have
\begin{align*}
H \left( X_{V_{W>0}} \right) & =  -\frac{1}{L} \sum_{i,j = 1} ^L c_j e^{- \left| i - \alpha_j \right| / \xi_j} \ln  \left[   c_j e^{- \left| i - \alpha_j \right| / \xi_j} \right] \\
& = -\frac{1}{L}  \sum_j \left( \ln  c_j - c_j \sum_{i} e^{- \left| i - \alpha_j \right| / \xi_j} \frac{ \left| i - \alpha_j \right| }{ \xi_j} \right) \,\;.
\end{align*}
The sum $\sum_i$ for $L \gg 1$ is
\begin{gather*}
 \sum_{i=1}^L e^{- \left| i - \alpha_j \right| / \xi_j} \frac{ \left| i - \alpha_j \right| }{ \xi_j}  \approx 2 \sum_{x=1}^\infty e^{- x / \xi_j} \frac{ x }{ \xi_j} \\
  = -\frac{2}{\xi_j} \frac{d}{d(\xi_j)^{-1}}\sum_{x=1}^\infty e^{- x / \xi_j} = 2\frac{e^{1/\xi_j}}{\left( e^{1/\xi_j} - 1 \right)^2 \xi_j} \,\;.
\end{gather*}
Using Eq.~\eqref{Eq_normalization_cj} together with the above, we get to the desired form \eqref{Eq_Heuristic_2}.
%\begin{multline*}
%H \left( X_{V_{W>0}} \right) =  \\  \frac{1}{L} \sum_{j = 1}^L \left( \left[ \xi_j \sinh (1/\xi_j)  \right]^{-1} - \log \left( \tanh\left[ (2 \xi_j)^{-1} \right) \right]  \right) \,\;.
%\end{multline*}

%On physical grounds, a value $0<W\ll 1$ should imply large lengths $\xi_j$, hence the above expression predicts a logarithmic divergence 
%\begin{align}
%\lim_{L \to \infty} \braket{ C^{(\rel)}_{B}(\mathcal V_{W>0})  }_\epsilon \cong - \frac{1}{L} \sum_j \log(\xi_j) \,\;
%\end{align}
%in agreement with the numerical results depicted in \autoref{Fig_entropy_fit}.

%\begin{figure}[ht]
%\centering
%\includegraphics[width=0.45\textwidth]{entropy_fit.pdf}
%\caption{Plot of $\lim_{L \to \infty} \braket{ C^{(\rel)}_{B}(\mathcal V_{W>0})  }_\epsilon$ for $L \to \infty$ as a function of the disorder strength $W$. The fitted curve is of the form $f(W) = - a \log(b W)$, where $a \approx -1.58$ and $b \approx 7.28 \cdot 10^{-2}$ are fitting parameters. \geo{I'm not sure if we should get rid of this plot}} \label{Fig_entropy_fit}
%\end{figure}

\section{Evaluation of Eq.~\eqref{Eq_Heuristic_1} for on-site energies following Cauchy distribution} \label{Sec_Cauchy}

\begin{figure}[t]
\centering
\includegraphics[width=0.45\textwidth]{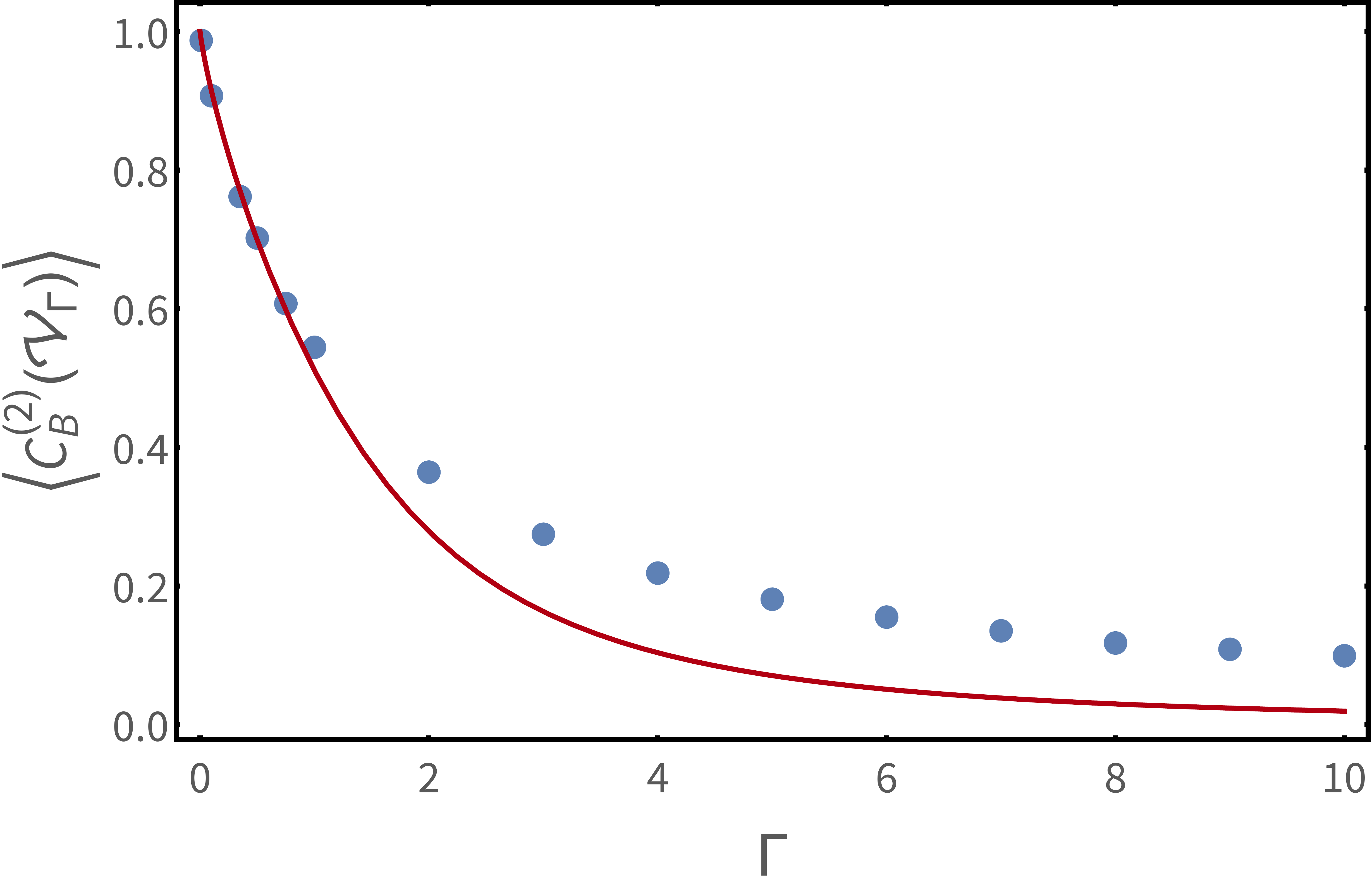}
\caption{Plot of the escape probability $\braket{C_B^{(2)}\left( \mathcal V_{\Gamma}  \right) }$ as a function of the disorder strength $\Gamma$ for the Lloyd model Hamiltonian $H_\Gamma$, as predicted analytically by the heuristic Eq.~\eqref{Eq_Heuristic_1} (solid line) and the numerical simulations (points). For the case of the numerical simulation, $L \to \infty$ is extrapolated by averaging over disorder for sizes up to $L = 2^{12}$. Standard deviations are within the point radius.}  \label{Fig_Lloyd}
\end{figure}

We consider the Hamiltonian \eqref{Eq_Anderson_Hamiltonian} with i.i.d. on-site energies $\epsilon_i$, distributed according to the Cauchy distribution
\begin{align}
f_\Gamma (\epsilon) = \frac{1}{\pi \Gamma} \left[ \frac{\Gamma^2}{\epsilon^2 + \Gamma^2}  \right] \,\;. 
\end{align}
The localization length $\xi(E,\Gamma)$ can be calculated by invoking the formula due to Thouless \cite{Thouless_1972}, which in our notation is
\begin{align}
\cosh\left( \frac{1}{2 \xi(E,\Gamma)} \right) = \frac{\sqrt{(2+E)^2+\Gamma^2} + \sqrt{(2-E)^2+\Gamma^2}}{4} \,. \label{Eq_Thouless}
\end{align}
To evaluate Eq.~\eqref{Eq_Heuristic_1} for this model in the thermodynamic limit, we transition to the continuum limit $\frac{1}{L}\sum_j g(E_j) \mapsto \int dE \rho_\Gamma(E) g(E)$. The density of states $\rho_\Gamma(E)$ can be obtained easily from the corresponding resolvent, calculated for the Lloyd model in \cite{Lloyd_1969}, and Eq.~\eqref{Eq_Thouless}. The resulting integral is numerically evaluated and yields the data plotted in \autoref{Fig_Lloyd}.

\section{Comparison of $C_B^{(2)}$ and $ f_B^{(\det)} $} \label{Sec_App_comparison_lambdas}

In this section, we will show that
\begin{align} \label{Eq_app_lambdas}
    P_{\text{return}} = 1 - C_B^{(2)} (\mathcal V) \ge \left(1 - f_B^{(\det)} (X_{\mathcal V}) \right)^2 \,\;.
\end{align}
Indeed,
\begin{gather*}
    1 - C_B^{(2)} (\mathcal V) = \frac{1}{d} \left\|  X_{\mathcal V} \right \|_2^2  = \frac{1}{d} \sum_i s_i^2 \\ 
     \ge \left( \frac{1}{d} \sum_i s_i \right) ^2  \ge \left[ \left( \prod_i s_i  \right)^{\frac{1}{d}}   \right]^2  =  \left( 1 - f_B^{(\det)} (X_{\mathcal V})  \right)^2 \,\;,
\end{gather*}
where $ s_i $ denotes the singular values of $X_{\mathcal V}$. The first equality follows from the convexity of the mean and the second one from the standard inequality between the arithmetic and geometric mean.

The inequality for the rates Eq.~\eqref{Eq_comparison_lambdas} follows by plugging into the inequality \eqref{Eq_app_lambdas} the forms \eqref{Eq_rate_2} and \eqref{Eq_rate_det}.

\section{Details of the numerical calculations for MBL}
\label{App_numerics}

\begin{figure}[t]
\subfloat[average escape probability (MBL)]{%
  \includegraphics[clip,width=\columnwidth]{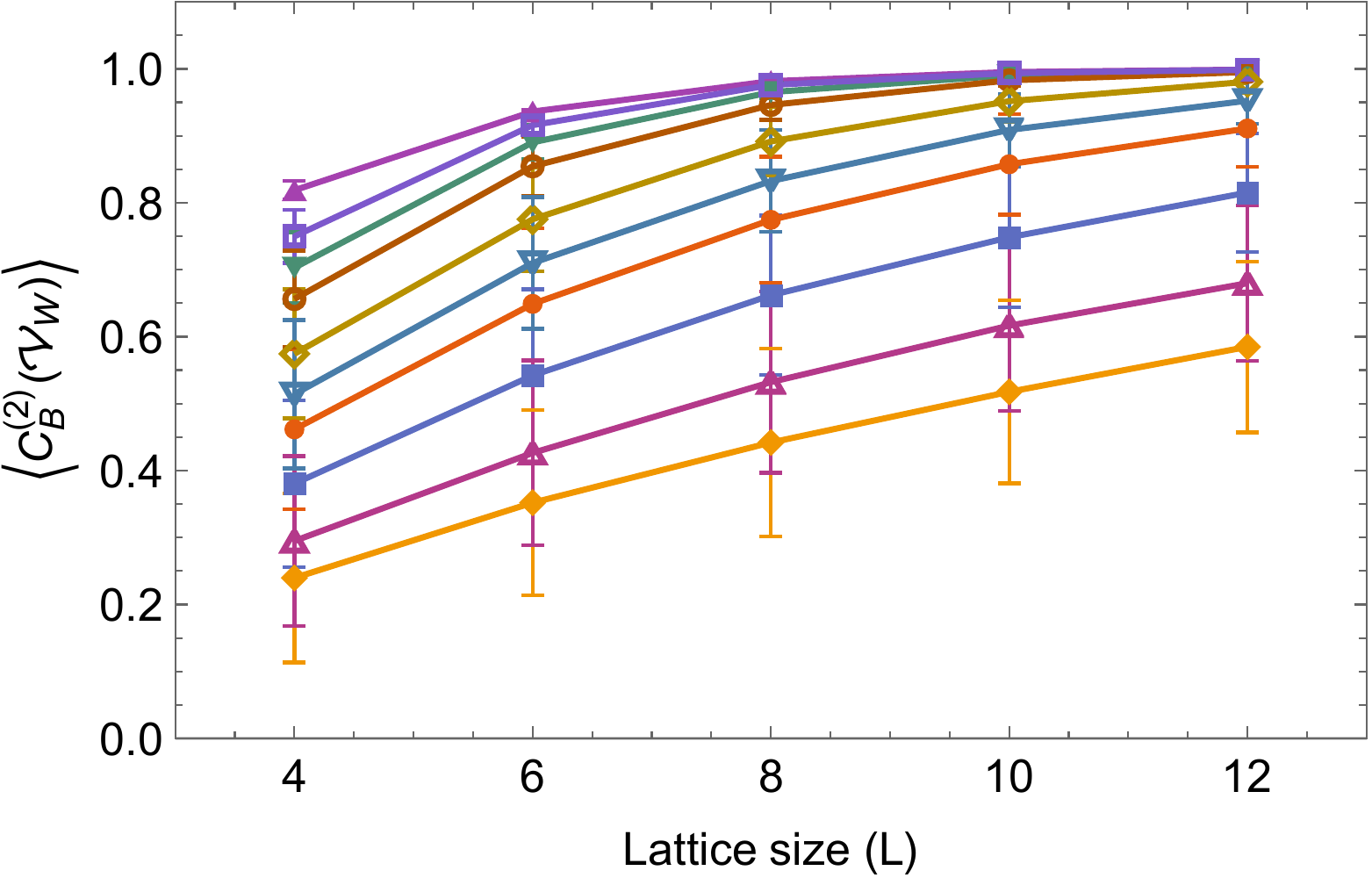}%
}

\subfloat[relative entropy CGP (MBL)]{%
  \includegraphics[clip,width=.94 \columnwidth]{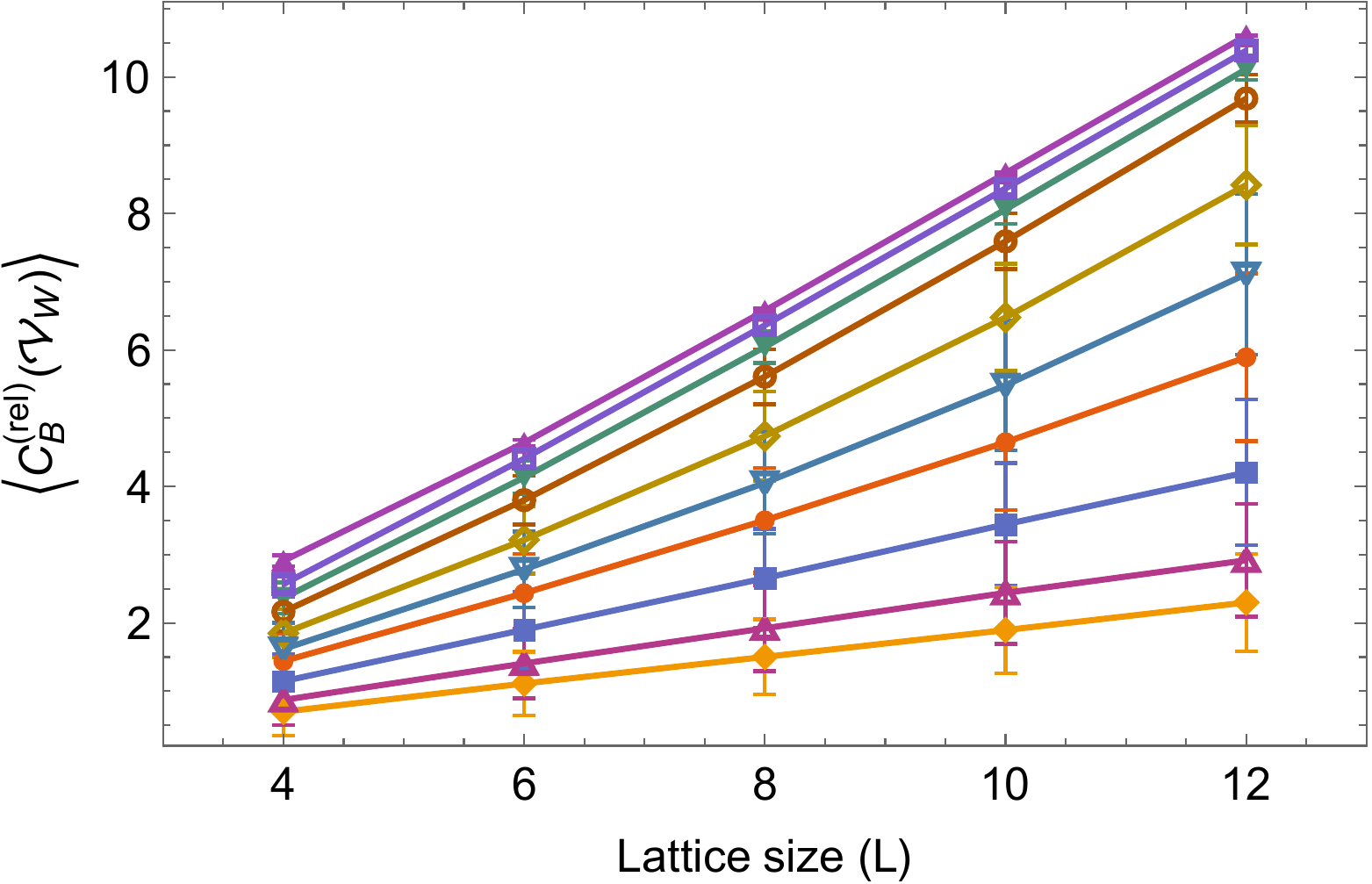}%
}

\caption{Plot of the \textbf{(a)}  average escape probability $\braket{ C_B^{(2)} \left( \mathcal V_W \right) }$ and \textbf{(b)} $\braket{ C_B^{(\rel)} \left( \mathcal V_W \right) }$ as a function of the system size $L$ for different values of the disorder strength $W$. The disorder values displayed are $W = 0.4, 1.0, 1.4, 1.8, 2.5, 3.1, 3.7, 5.0, 7.0, 9.0$ (monotonically from the top to bottom in the plots) for $L=4,6,\cdots, 12$  with sample sizes $\num{20000}, \num{20000}, \num{20000}, \num{8000}, \num{2000}$; except at the $W=3.7$, where the sample sizes were doubled. Error bars represent one standard deviation. Entropy has logarithm with base 2.} 
\label{MBL_details1}
\end{figure}

\begin{figure}[!ht]
\subfloat[time-averaged CGP (MBL)]{%
  \includegraphics[clip,width=.94 \columnwidth]{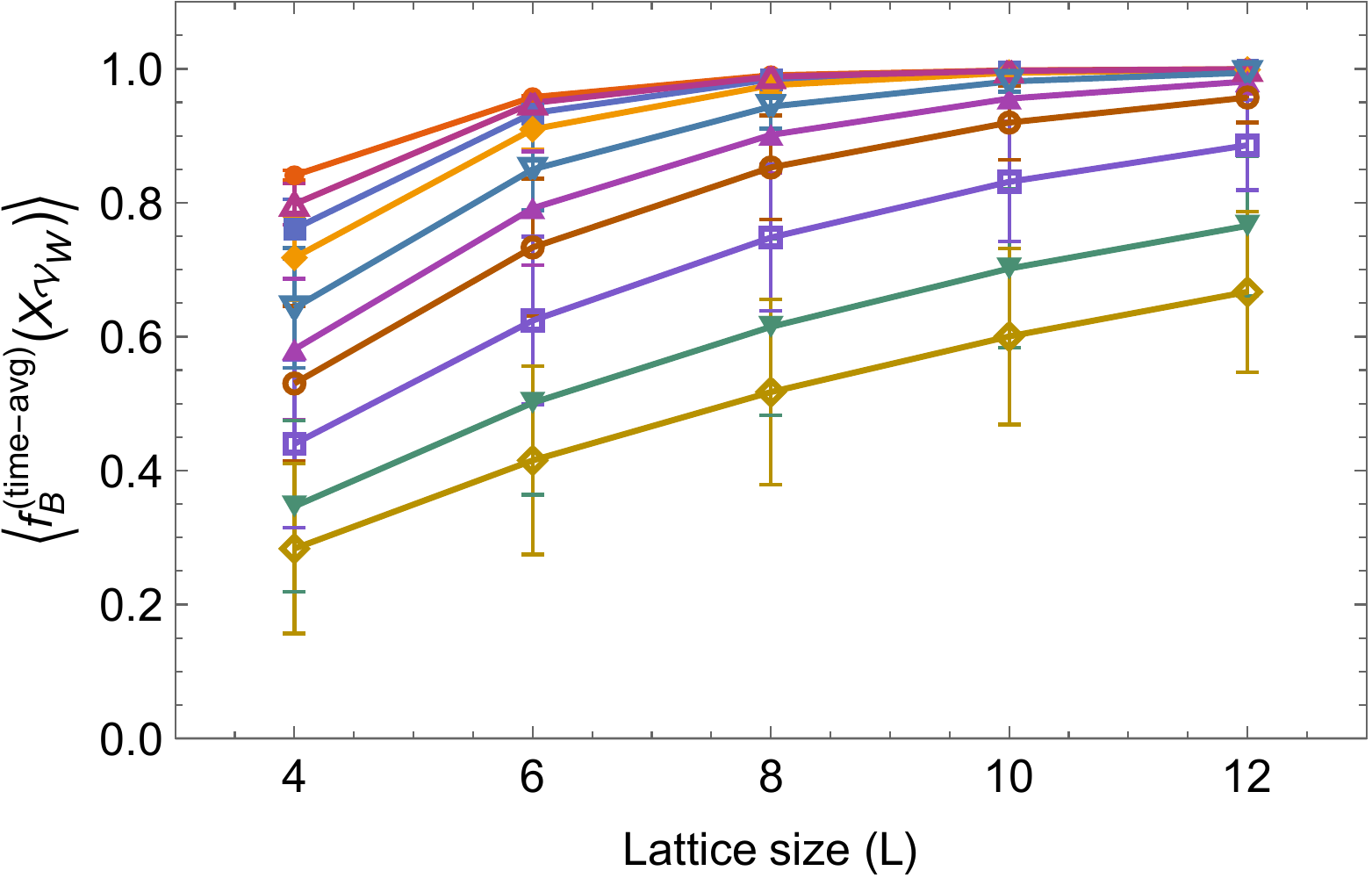}%
}

\subfloat[$f_B^{(\det)}$ CGP (MBL)]{%
  \includegraphics[clip,width=.94 \columnwidth]{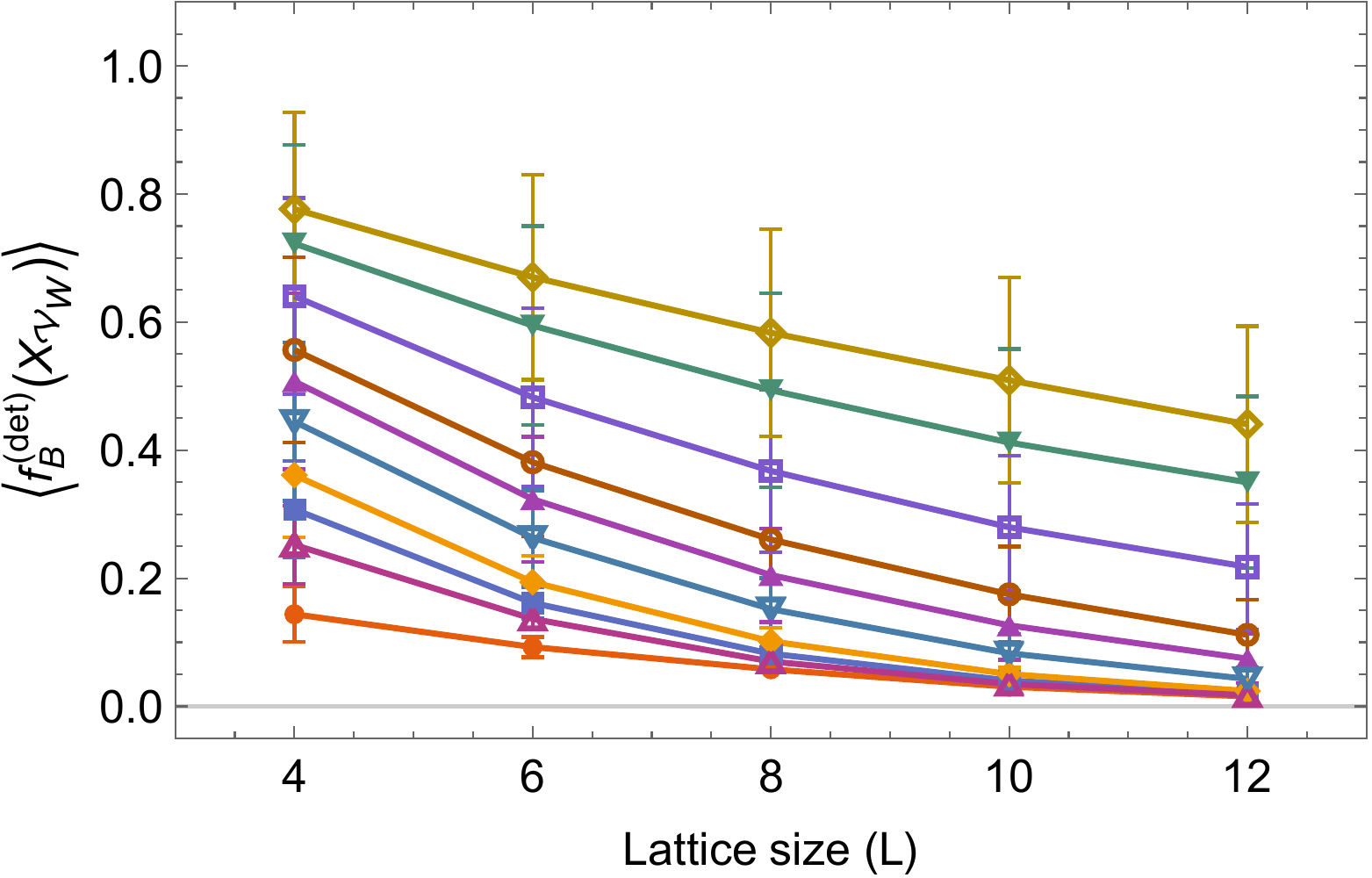}%
}

\subfloat[$f_B^{(\infty)}$ CGP (MBL)]{%
  \includegraphics[clip,width=.94 \columnwidth]{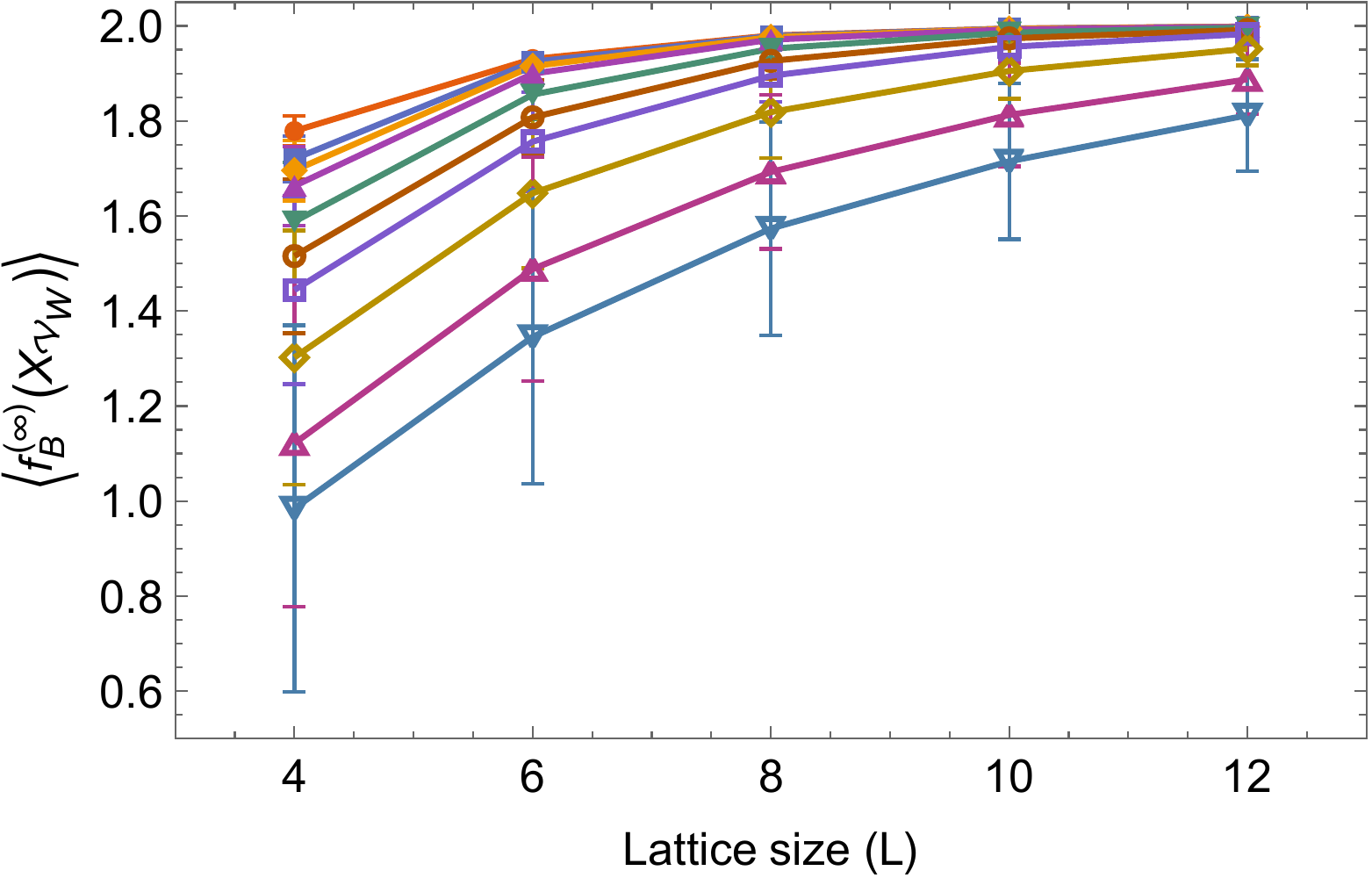}%
}

\caption{Plot of the generalized-CGP measures, \textbf{(a)} $\braket{f_B^{(\text{time-avg})} \left( X_{ \mathcal V_W }\right) }$, \textbf{(b)} $\braket{f_B^{(\text{det})} \left( X_{ \mathcal V_W }\right) }$, and \textbf{(c)} $\braket{f_B^{(\infty)} \left( X_{ \mathcal V_W }\right) }$ as a function of the system size $L$ for different values of the disorder strength $W$. The disorder values displayed are $W = 0.4, 1.0, 1.4, 1.8, 2.5, 3.1, 3.7, 5.0, 7.0, 9.0$ (monotonically from the top to bottom for (a) and bottom to top for (b)) for $L=4,6,\cdots, 12$ with sample sizes $\num{20000}, \num{20000}, \num{20000}, \num{8000}, \num{2000}$; except at the $W=3.7$, where the sample sizes were doubled. Error bars represent one standard deviation.} 
\label{MBL_details2}
\end{figure}

In this section, we list further details of the
%CGP-based
quantities studied across the ergodic-MBL transition, namely $1 -\braket{ C_B^{(2)} \left( \mathcal V_W \right) }$, $\braket{ C_B^{(\rel)} \left( \mathcal V_W \right) }$, $\braket{f_B^{(\text{time-avg})} \left( X_{ \mathcal V_W }\right) }$, and $\braket{f_B^{(\text{det})} \left( X_{ \mathcal V_W }\right) }$. In Figure~\ref{Fig_MBL}, we plot the extrapolated rates (for large $L$) as a function of the disorder strength $W$ for the Hamiltonian $H_{\mathrm{XXX}}$ at $h_x=0.3$. For this purpose we consider, e.g., for the return probability $1 -\braket{ C_B^{(2)} \left( \mathcal V_W \right) }$ an ansatz of the form
\begin{equation}
    g(L) = \alpha + 2^{- \lambda L} \,\;,
\end{equation}
where $\alpha$ is the asymptotic value and $\lambda$ is the rate of decay with system size $L$. By performing a nonlinear fit at different disorder values for the various quantities listed above, we found that the $\alpha \approx 0$
%(of the order of $10^{-2}$)
(within the uncertainty of the fitting parameters), even for the largest disorder that we consider ($W=9.0$). Therefore, we simplify our ansatz to the form $g(L) \propto 2^{- \lambda L}$ and extract the asymptotic rates by taking the logarithm of the desired quantities. 

In Figures~\ref{MBL_details1} and \ref{MBL_details2} we plot our data for a sample of disorder values and
%for illustration purposes (plotting all possible disorder values would make the graphs unreadable)
for system sizes $L=4,\cdots, 12$.
%Note, however, that in Figure~\ref{Fig_MBL}, we plot exponents for many more disorder values and $L$ up to $14$.
Error bars represent one standard deviation.

\section{Coherence-Generating Power and distance in the Grassmannian} \label{sec_diffgeomsupp}

Here we present in more detail the underlying differential-geometric structure that is introduced in \autoref{sec_diffgeom}.

Let $\cal H$ denote the finite dimensional Hilbert space of the quantum system and ${\mathcal B}({\mathcal H})$ the associate operator algebra. The set  ${\mathcal B}({\mathcal H})$ equipped with the Hilbert-Schmidt scalar product $\braket{X,Y} \coloneqq \Tr \left( X^\dagger Y \right)$ turns into a Hilbert space (the space of Hilbert-Schmidt operators) that we will denote by ${\mathcal H}_{\HS}$. Superoperators $\cal O$ mapping ${\mathcal H}_{\HS}$ into itself can be then endowed with the following norm
\begin{align}
\|{\cal O}\|_{\HS}  \coloneqq \sqrt{ {\mathrm{Tr}}_{\HS} ({\mathcal O}^\dagger {\mathcal O})} \,\;,
\end{align}
where
\begin{inparaenum}[(a)]
\item ${\cal O}^\dagger $ denotes the Hilbert-Schmidt conjugate of $\mathcal O$, i.e., $\braket{\mathcal O (X),Y} = \braket{ X, \mathcal O ^\dagger (Y)}$ $\forall \,X,Y \in \mathcal H_{\HS}$.
\item If $\left\{  \ket{i} \right\}_{i=1}^{d}$ is any orthonormal basis of $\mathcal H$, one defines ${\mathrm{Tr}}_{\HS}{\mathcal O} \coloneqq \sum_{i,j=1}^{d} \braket{ \ket{i}\!\bra{j}, \mathcal O (\ket{i}\!\bra{j})}$. 
\end{inparaenum}

As we discussed in the main text, instead of invoking orthonormal sequences of kets $\left\{ \ket{i} \right\}_{i=1}^{d}$, it is more convenient to work with sets of orthogonal, rank-1 projection operators $B=\left\{P_i \coloneqq  \ket{i}\!\bra{i} \right\}_{i=1}^d$. Let us introduce the space of all such sets over the Hilbert space,  which we denote as $\mathcal M(\mathcal H)$. This is essentially the set of all possible orthonormal bases over the Hilbert space once the phase degrees of freedom and ordering have been modded out \cite{zanardi2018quantum}. The elements $B \in \mathcal M(\mathcal H)$ are in one-to-one correspondence with the set of dephasing super-operators, i.e., the map $B \mapsto \mathcal D_B$ (defined in Eq.~\eqref{Eq_dephasing_def}) is injective.
Given a $B\in{\cal M}_d$, the corresponding set of $B$-diagonal operators is
\begin{align}
{\mathcal A}_B \coloneqq  \Span \left\{P_i\right\}_{i=1}^{d}  \subset {\mathcal H}_{\HS} \,\;,
\label{Eq_A_B}
\end{align}
which is also the range of the $B$-dephasing superoperator ${\cal D}_B.$
One can see that Eq.~\eqref{Eq_A_B} actually defines a Maximally Abelian Sub-Algebra (MASA) of ${\mathcal H}_{\HS}$; moreover 
it can be proven that the set of MASAs of $\mathcal H_{\HS}$ can be identified with  ${\mathcal M}(\mathcal H)$ (see \cite{zanardi2018quantum} for a proof). In this way, the set ${\mathcal M}(\mathcal H)$ can be now seen as 
a subset of the Grassmannian manifold of $d$-dimensional subspaces of ${\mathcal H}_{\HS}$. The advantage of this approach is that $\mathcal M (\mathcal H)$ directly inherits the natural metric structure of the Grassmannian 
\begin{align}
D({\mathcal A}_B, {\mathcal A}_{B'}) \coloneqq  \left\|{\mathcal D}_B - {\mathcal D}_{B'}\right\|_{\HS} \,\;.
\label{MASA-dist}
\end{align}
We will now connect these concepts to the 2-CGP of unitary quantum maps.

From its definition, $C_{B}^{(2)} (\mathcal U)$ seems to capture some notion of separation between the sets $B = \left\{  P_i \right\}_{i=1}^{d}$ and $B' = \left\{  \mathcal U \left(  P_i \right) \right\}_{i=1}^{d}$. In fact, the $B$-coherence generating power of a unitary map $\mathcal U$ is proportional to the (square of the) Grassmannian distance between the input $B$-diagonal algebra ${\cal A}_B$ and its image under $\mathcal U$ \cite{zanardi2018quantum}. Formally:
\begin{align} 
C^{(2)}_{B}({\mathcal U})=\frac{1}{2 d} D({\mathcal A}_B, {\cal U}({\mathcal A}_{{B}}))^2.
\label{C_B-D}
\end{align}
where the distance function $D$ is given by (\ref{MASA-dist}). The maximum of this function i.e., $\max_{\mathcal U} C_B^{(2)} (\mathcal U)=1 - 1/d$ is achieved for unitary operators $\mathcal U$ that connected mutually unbiased bases, namely $\left|\braket{i|U|j} \right| = 1/d$ ($\forall i,j$), and corresponds to a maximum distance over ${\mathcal M}(\mathcal H)$ given by $D^{\max}=\sqrt{2(d-1)}.$ 

It is important to stress that, in the light of \autoref{Prop_CGP_meaning}, the Grassmannian distance between MASAs is endowed with a physical meaning in the context of quantum mechanics.

We now turn to establish a connection between the differential structure of $\mathcal M (\mathcal H)$, as induced by the distance function \eqref{MASA-dist}, and MBL. 
%We now  ready to consider a differentiable metric structure over ${\mathcal M}(\mathcal H)$.
%and will argue that it has a physical content.
One has the natural Riemannian metric over the Grassmannian
\begin{align}
ds^2 = D(\Pi, \Pi + d\Pi)^2 = \Tr(d\Pi^2)
\end{align}
($\Pi$ denote the projectors over the $d$-dimensional subspaces comprising the Grassmannian).  The latter, in view of Eq.~\eqref{C_B-D}, has in turn the physical interpretation as the $C_B^{(2)}$ of the unitary associated with an infinitesimal transformation $\left\{\ket{\phi_i(\lambda)}\right\}_{i=1}^{d} \mapsto \left\{\ket{\phi_i(\lambda+ d\lambda)}\right\}_{i=1}^{d}$. The form of the metric \eqref{Eq_metric} follows directly by the calculation of Proposition 6 in \cite{zanardi2018quantum}.

\bibliographystyle{apsrev4-1}
\bibliography{Localization,refs}

\end{document}